\documentclass[letter, 11pt]{article}

\newif\ifdraft
\drafttrue

\usepackage{dsfont}
\usepackage{amsmath}
\usepackage{amsthm}
\usepackage{amssymb}
\usepackage{algorithmicx}
\usepackage{algorithm}
\usepackage{algpseudocode}
\usepackage{graphicx}
\usepackage{url}
\usepackage{mathabx}
\usepackage[letterpaper,margin=1in]{geometry}
\usepackage{authblk}
\usepackage{subcaption}
\usepackage{tikz}
\usepackage{xspace}

\newcommand{\R}{\mathbb{R}}    
\theoremstyle{plain}
\newtheorem{theorem}{Theorem}[section]
\newtheorem{lemma}[theorem]{Lemma}
\newtheorem{definition}[theorem]{Definition}
\newtheorem{defn}[theorem]{Definition}

\newtheorem{proposition}[theorem]{Proposition}
\newtheorem{remark}[theorem]{Remark}

\newtheorem{fact}[theorem]{Fact}
\newtheorem{claim}[theorem]{Claim}

\newcommand{\etc}{\textit{etc.}\xspace}
\newcommand{\reals}{\mathbb{R}}
\newcommand{\naturals}{\mathbb{N}}
\newcommand{\calC}{\mathcal{C}}

\newcommand{\calD}{\mathcal{D}}
\newcommand{\calI}{\mathcal{I}}
\newcommand{\calN}{\mathcal{N}}

\newcommand{\calW}{\mathcal{W}}
\newcommand{\calV}{\mathcal{V}}
\usepackage{hyperref}

\theoremstyle{remark}

\newtheorem*{note}{Note}

\renewcommand{\th}{^{\footnotesize{\it th}}}

\newcommand{\tT}{\tilde{T}}

\newcommand{\eps}{\varepsilon}
\newcommand{\lca}{\text{LCA}}
\newcommand{\opt}{\text{OPT}}
\newcommand{\OPT}{\text{OPT}}

\newcommand{\dist}{\text{dist}}

\newcommand{\cut}{w}

\newcommand{\cost}{\text{cost}}
\newcommand{\val}{\text{val}}

\renewcommand{\tilde}{\widetilde}
\newcommand{\twonorm}[1]{\Vert #1 \Vert_{\mbox{\normalfont\scriptsize 2}}}
\renewcommand{\Pr}[1]{\mathbb{P}\left[\,#1\,\right]}
\newcommand\E[1]{\mathbb{E}\left[\,#1\,\right]}

\usepackage[round]{natbib}

\newcommand{\structinput}{ground-truth input\xspace}
\newcommand{\structinputs}{ground-truth inputs\xspace}

\newcommand{\StructInputs}{Ground-Truth Inputs\xspace}

\newcommand{\tvi}{\tilde{V_i}}

\newcommand{\hy}{\hat{y}}

\def\eg{{\it e.g.,}\xspace}
\def\ie{{\it i.e.,}\xspace}

\usepackage{color}

\definecolor{blueblack}{rgb}{0,0,.7}

\newcounter{sideremark}
\newcommand{\marrow}{\stepcounter{sideremark}\marginpar{$\boldsymbol{\longleftarrow\scriptstyle\arabic{sideremark}}$}}

\newcommand{\vknote}[1]{
\ifdraft
\textcolor{red}{VK: \marrow #1}
\fi}

\newcommand{\costnode}{\gamma}
\newcommand{\Costtree}{\Gamma}

\newcommand{\cho}{\text{child}_1}
\newcommand{\choc}[1]{\text{child}_1(#1)}

\newcommand{\cht}{\text{child}_2}
\newcommand{\chtc}[1]{\text{child}_2(#1)}

\makeatletter
\newtheorem*{rep@theorem}{\rep@title}
\newcommand{\newreptheorem}[2]{%
  \newenvironment{rep#1}[1]{%
    \def\rep@title{\textbf{\emph{#2 \ref{##1}}}}%
    \begin{rep@theorem}}%
    {\end{rep@theorem}}}
\makeatother

\newreptheorem{theorem}{Theorem}
\newreptheorem{thm}{Theorem}
\newreptheorem{lem}{Lemma}
\newreptheorem{defn}{Definition}
\newreptheorem{lemma}{Lemma}
\newreptheorem{proposition}{Proposition}

\newreptheorem{algorithm}{Algorithm}

\newif\ifshort
  \shortfalse 

\title{Hierarchical Clustering: Objective Functions and Algorithms}
\author[1]{Vincent Cohen-Addad}
\author[2,3]{Varun Kanade}
\author[4,5]{Frederik Mallmann-Trenn}
\author[4]{Claire Mathieu}
\affil[1]{University of Copenhagen}
\affil[2]{University of Oxford}
\affil[3]{The Alan Turing Institute}
\affil[4]{\'{E}cole normale sup\'{e}rieure, CNRS, {\tiny PSL Research University}}
\affil[5]{Simon Fraser University}
\date{}
\begin{document}
\begin{titlepage}
  \maketitle
  \thispagestyle{empty}
\begin{abstract}
	Hierarchical clustering is a recursive partitioning of a dataset into
	clusters at an increasingly finer granularity.  Motivated by the fact
	that most work on hierarchical clustering was based on providing
	algorithms, rather than optimizing a specific objective,
	\cite{Das:2016} framed similarity-based hierarchical clustering as a
	combinatorial optimization problem, where a `good' hierarchical
	clustering is one that minimizes some cost function. He showed that
	this cost function has certain desirable properties, such as in order
	to achieve optimal cost disconnected components must be separated
	first and that in `structureless' graphs, \ie cliques, all clusterings
	achieve the same cost.  
	
	We take an axiomatic approach to defining `good' objective functions
	for both similarity and dissimilarity-based hierarchical clustering.
	We characterize a set of \emph{admissible} objective functions (that
	includes the one introduced by Dasgupta) that have the property that
	when the input admits a `natural' ground-truth hierarchical
	clustering, the ground-truth clustering has an optimal value.

	Equipped with a suitable objective function, we analyze the performance of
	practical algorithms, as well as develop better and faster algorithms for
	hierarchical clustering. For similarity-based hierarchical clustering,
	\cite{Das:2016} showed that a simple recursive sparsest-cut based approach
	achieves an $O(\log^{3/2} n)$-approximation on worst-case inputs.  We give a
	more refined analysis of the algorithm and show that it in fact achieves an
	$O(\sqrt{\log n})$-approximation\footnote{\cite{CC16} independently proved
	that the sparsest-cut based approach achieves a $O(\sqrt{\log n})$
	approximation.}. This improves upon the LP-based $O(\log n)$-approximation
	of~\cite{Roy16}. For dissimilarity-based hierarchical clustering, we show
	that the classic average-linkage algorithm gives a factor $2$ approximation,
	and provide a simple and better algorithm that gives a factor $3/2$
	approximation. This aims at explaining the success of this heuristics in
        practice.
	%
	Finally, we consider `beyond-worst-case' scenario through 
        a generalisation of the stochastic block model for
	hierarchical clustering. We show that Dasgupta's cost function also has
	desirable properties for these inputs and we provide a simple algorithm
	that for graphs generated according to this model yields a 1 + o(1) factor
	approximation.

\end{abstract}
\end{titlepage}
\setcounter{page}{1}

\newpage
\tableofcontents 
\bigskip

  

\newpage

\section{Introduction}
\label{S:intro}
A \emph{hierarchical clustering} is a recursive partitioning of a dataset
into successively smaller clusters.  The input is a weighted graph whose
edge weights represent pairwise similarities or dissimilarities between
datapoints. A hierarchical clustering is 
represented by a rooted tree where each leaf represents a datapoint
and each internal node represents a cluster containing its
descendant leaves. Computing a hierarchical clustering is a fundamental
problem in data analysis; it is routinely used to analyze, classify, and
pre-process large datasets. A hierarchical clustering provides useful
information about data that can be used, \eg to divide a digital image
into distinct regions of different granularities, to identify communities
in social networks at various societal levels, or to determine the
ancestral tree of life. Developing robust and efficient algorithms for
computing hierarchical clusterings is of importance in several research
areas, such as machine learning, big-data analysis, and bioinformatics.

Compared to flat {partition-based clustering} (the problem 
of dividing the dataset into $k$ parts), hierarchical clustering has
received significantly less attention from a theory perspective.
Partition-based clustering is typically framed as minimizing a
well-defined objective such as $k$-means, $k$-medians, \etc and
(approximation) algorithms to optimize these objectives have been a focus
of study for at least two decades. On the other hand, hierarchical
clustering has rather been studied at a more procedural level in terms of
algorithms used in practice. Such algorithms can be broadly classified
into two categories, \emph{agglomerative} heuristics which build the
candidate cluster tree bottom up, \eg average-linkage, single-linkage,
and complete-linkage, and \emph{divisive} heuristics which build the tree
top-down, \eg bisection $k$-means, recursive sparsest-cut \etc
\citet{Das:2016} identified the lack of a well-defined objective function
as one of the reasons why the theoretical study of hierarchical
clustering has lagged behind that of partition-based clustering.

\subsubsection*{Defining a Good Objective Function.} 

What is a `good' output tree for hierarchical clustering?  Let us suppose that
the edge weights represent similarities
(similar datapoints are connected by edges of high weight)\footnote{This
entire discussion can equivalently be phrased in terms of dissimilarities
without changing the essence.}. \citet{Das:2016} frames
hierarchical clustering as a combinatorial optimization problem, where a good
output tree is a tree that minimizes some cost function; but which function
should that be? Each (binary) tree node is naturally associated to a cut that
splits the cluster of its descendant leaves into the cluster of its left
subtree on one side and the cluster of its right subtree on the other, and
Dasgupta defines the objective to be the sum, over all tree nodes, of the total
weight of edges crossing the cut multiplied by the cardinality of the node's
cluster.  In what sense is this good?  Dasgupta argues that it has several attractive properties:  (1) if the graph is
disconnected, \ie data items in different connected components have nothing to
do with one another, then the hierarchical clustering that minimizes the
objective function begins by first pulling apart the connected components from
one another; (2) when the input is a (unit-weight) clique then no particular
structure is favored and all binary trees have the same cost; and (3) the cost
function also behaves in a desirable manner for data containing a planted
partition. Finally, an attempt to generalize the cost function leads to
functions that violate property (2). 

In this paper, we take an axiomatic approach to defining a `good' cost
function.  We remark that in many application, for example in
phylogenetics, there exists an unknown `ground truth' hierarchical
clustering--- the actual ancestral tree of life---from which the similarities
are generated (possibly with noise), and the goal is to infer the underlying
ground truth tree from the available data. In this sense, a cluster tree is
good insofar as it is isomorphic to the (unknown) ground-truth cluster tree,
and thus a natural condition for a `good' objective function is one such that
for inputs that admit a `natural'  ground-truth cluster tree, the value of the
ground-truth tree is optimal. We provide a formal definition of inputs that
admit a ground-truth cluster tree in Section~\ref{SubSec:ultrametric}.



We consider, as potential objective functions, the class of all functions
that sum, over all the nodes of the tree, the total weight of edges
crossing the associated cut times some function of the cardinalities of
the left and right clusters (this includes the class of functions
considered by~\cite{Das:2016}).  In Section~\ref{sec:cost-functions} we
characterize the `good' objective functions in this class and call them
\emph{admissible} objective functions. We prove that for any objective
function, for any \structinput, the ground-truth tree has optimal cost
(w.r.t to the objective function) \emph{if and only if} the objective
function (1) is symmetric (independent of the left-right order of
children), (2) is increasing in the cardinalities of the child clusters,
and (3) for (unit-weight) cliques, has the same cost for all binary trees
(Theorem~\ref{thm:cost-func}).   Dasgupta's objective
function is admissible in terms of the criteria described above. 

In Section~\ref{S:randominputs}, we consider random graphs that induce a
natural clustering. This model can be seen as a noisy version
of our notion of ground-truth inputs and a hierarchical stochastic block 
model. We show that the
ground-truth tree has optimal \emph{expected cost} for any admissible
objective function. Furthermore, we show that the ground-truth tree has
cost at most $(1+o(1)) \opt$ with high probability for the objective
function introduced by~\cite{Das:2016}.


\subsection*{Algorithmic Results}

The objective functions identified in Section~\ref{S:costfun} allow us to (1)
quantitatively compare the performances of algorithms used in practice and (2)
design better and faster approximation algorithms.%
\footnote{For the objective function proposed in his work, \citet{Das:2016}
shows that finding a cluster tree that minimizes the cost function is NP-hard.
This directly applies to the admissible objective functions for the
dissimilarity setting as well. Thus, the focus turns to developing
approximation algorithms.} 
\smallskip

\noindent{\bf Algorithms for Similarity Graphs}:
\citet{Das:2016} shows that the \emph{recursive $\phi$-approximate
sparsest cut} algorithm, that recursively splits the input graph using a
$\phi$-approximation to the sparsest cut problem, outputs a tree whose cost is
at most $O(\phi \log n \cdot \OPT)$. \cite{Roy16} recently gave an $O(\log
n)$-approximation by providing a linear programming relaxation for the problem
and providing a clever rounding technique.  \cite{CC16} showed that the
recursive $\phi$-sparsest cut algorithm of Dasgupta gives an
$O(\phi)$-approximation.  In Section~\ref{S:sim:worstcase}, we obtain an
independent proof showing that the $\phi$-approximate sparsest cut algorithm is
an $O(\phi)$-approximation (Theorem~\ref{thm:SCapprox})\footnote{Our analysis
shows that the algorithm achieves a $6.75\phi$-approximation and the analysis
of~\cite{CC16} yields a $8\phi$-approximation guarantee.  This minor difference
is of limited impact since the best approximation guarantee for sparsest-cut is
$O(\sqrt{\log n})$.  }. Our proof is quite different from the proof
of~\cite{CC16} and relies on a charging argument.  Combined with the celebrated
result of~\cite{AroraRV09}, this yields an $O(\sqrt{\log n})$-approximation.
The results stated here apply to Dasgupta's objective function; the
approximation algorithms extend to other objective functions, though the ratio
depends on the specific function being used.   We conclude our analysis of the
worst-case setting by showing that all the linkage-based algorithms commonly
used in practice can perform rather poorly on worst-case inputs (see
Sec.~\ref{ssec:wc-lower-bounds}). \smallskip

\noindent{\bf Algorithms for Dissimilarity Graphs}:
Many of the algorithms commonly used in practice, \eg linkage-based methods,
assume that the input is provided in terms of pairwise dissimilarity (\eg
points that lie in a metric space).  As a result, it is of interest to
understand how they fare when compared using admissible objective functions for
the dissimilarity setting.  When the edge weights of the input graph represent
dissimilarities, the picture is considerably different from an approximation
perspective.  For the analogue of Dasgupta's objective function in the
dissimilarity setting, we show that the average-linkage algorithm (see
Algorithm~\ref{alg:avglinkage}) achieves a $2$-approximation
(Theorem~\ref{T:avgapprox}). This stands in contrast to other practical
heuristic-based algorithms, which may have an approximation guarantee as bad as
$\Omega(n^{1/4})$ (Theorem~\ref{T:badcase:dis:pract}). Thus, using this
objective-function based approach, one can conclude that the average-linkage
algorithm is the more robust of the practical algorithms, perhaps explaining
its success in practice.  We also provide a new, simple, and better algorithm,
the locally densest-cut algorithm,%
\footnote{We say that a cut $(A,B)$ is locally dense if moving a vertex from
$A$ to $B$ or from to $B$ to $A$ does not increase the density of the cut. One
could similarly define locally-sparsest-cut.}
which we show gives a $3/2$-approximation (Theorem~\ref{T:dis:betterapprox}).
Our results extend to any admissible objective function,
though the exact approximation factor depends on the specific choice. \smallskip

\noindent\textbf{Structured Inputs and Beyond-Worst-Case Analysis}: 
The recent work of~\cite{Roy16} and~\cite{CC16} have shown that obtaining
constant approximation guarantees for worst-case inputs is beyond current
techniques (see Section~\ref{ssec:related-work}).  Thus, we consider
inputs that admit a `natural' ground-truth cluster tree. For such inputs,
we show that essentially all the practical algorithms do the right thing,
in that they recover the ground-truth cluster tree. Since real-world
inputs might exhibit a noisy structure, we consider more general scenarios:

\begin{itemize}
	\item We consider a natural generalization of the classic stochastic block
		model that generates random graphs with a hidden ground-truth
		hierarchical clustering. We provide a simple algorithm based on singular
		value decomposition (SVD) and agglomerative methods that achieves a
		$(1+o(1))$-approximation for Dasgupta's objective function (in fact, it
		recovers the ground-truth tree) with high probability. Interestingly,
		this algorithm is very similar to approaches used in practice for
		hierarchical clustering.
  
\item We introduce the notion of a $\delta$-\emph{adversarially
    perturbed} ground-truth input, which can be viewed as being
  obtained from a small perturbation to an input that admits a
  natural ground truth cluster tree. This approach bears similarity
  to the stability-based conditions used by~\citet{BalcanBV10}
  and~\citet{BiL12}.  We provide an algorithm that achieves a
  $\delta$-approximation in both the similarity and dissimilarity
  settings, independent of the objective function used as long as it
  is admissible according to the criteria used in
  Section~\ref{sec:cost-functions}.
\end{itemize}

\subsection{Summary of Our Contributions}
\label{ssec:result-summary}

Our work makes significant progress towards providing a more complete
picture of objective-function based hierarchical clustering and understanding
the success of the classic heuristics for hierarchical clustering.

\begin{itemize}
	\item 
		Characterization of `good' objective functions. We prove
		that for any \structinput, the ground-truth tree has strictly
		optimal cost for an objective function {\em if and only if}, the
		objective function (1) is symmetric (independent of the left-right
		order of children), (2) is monotone in the cardinalities of the
		child clusters, and (3) for unit-weight cliques, gives the same
		weight to all binary trees (Theorem~\ref{thm:cost-func}).  We refer
		to such objective functions as \emph{admissible}; according to
		these criteria Dasgupta's objective function is admissible.
	\item Worst-case approximation. First, for similarity-based inputs, we provide a new proof that the
		recursive $\phi$-approximate sparsest cut algorithm is an
		$O(\phi)$-approximation (hence an $O(\sqrt{\log
		n})$-approximation) (Theorem~\ref{thm:SCapprox}) for Dasgupta's
		objective function.  Second, for dissimilarity-based inputs, we show that the
		classic average-linkage algorithm is a $2$-approximation
		(Theorem~\ref{T:avgapprox}), 
		 and provide a new
		algorithm which we prove is a $3/2$-approximation (Theorem~\ref{T:dis:betterapprox}).
		All those results extend to other cost functions
		but the approximation ratio is function-dependent. 
	\item Beyond worst-case. First, stochastic models. 
		We
		consider the \emph{hierarchical stochastic block model} 
                (Definition~\ref{defn:hsbm}). 
		We give a simple algorithm based on SVD and
		classic agglomerative methods that, with high probability,
		recovers the ground-truth tree and show that this tree has
		cost that is $(1 + o(1)) \opt$ with respect to Dasgupta's objective
		function (Theorem~\ref{thm:random:recovery}). Second, adversarial models. We introduce the notion of 
		$\delta$-{perturbed} inputs, obtained by
		a small adversarial perturbation to \structinputs, and give a simple 
		$\delta$-approximation algorithm 
		(Theorem~\ref{T:robustsimple}).  
		%

	\item Perfect inputs, perfect reconstruction. For \structinputs, 
		we note that the algorithms used in practice (the
		linkage algorithms, the bisection $2$-centers, etc.) correctly
		reconstruct a ground truth tree
		(Theorems~\ref{T:linkageAlgs:Perfectdata},~\ref{T:bisectionkcenter},
		\ref{T:sc:perfect}).  We introduce a simple, faster 
		algorithm 
		that is also optimal on \structinputs
		(Theorem~\ref{T:alg:fastsimple}).  
\end{itemize}

\subsection{Related Work}
\label{ssec:related-work}
The recent paper of~\citet{Das:2016} served as the starting point of this
work. \citet{Das:2016} defined an objective function for hierarchical
clustering and thus formulated the question of constructing a cluster
tree as a combinatorial optimization problem.  Dasgupta also showed that
the resulting problem is NP-hard and that the recursive
$\phi$-sparsest-cut algorithm achieves an $O(\phi \log n)$-approximation.
Dasgupta's results have been improved in two subsequent papers.
\cite{Roy16} wrote an integer program for the
hierarchical clustering problem using a combinatorial
characterization of the ultrametrics induced by Dasgupta's cost function.
They also provide a spreading metric LP
and a rounding algorithm based on sphere/region-growing that yields an $O(\log
n)$-approximation. Finally, they show that no polynomial size
SDP can achieve a constant factor approximation for the problem and that
under the Small Set Expansion (SSE) hypothesis, no polynomial-time
algorithm can achieve a constant factor approximation.

\citet{CC16} also gave a proof that the problem is hard to
approximate within any constant factor under the Small Set Expansion
hypothesis. They also proved that the recursive
$\phi$-sparsest cut algorithm produces a hierarchical clustering with
cost at most $O(\phi \OPT)$; 
their techniques appear to be significantly different
from ours. Additionally, \cite{CC16} introduce a spreading metric SDP
relaxation for the hierarchical clustering problem introduced by Dasgupta
that has integrality gap $O(\sqrt{\log n})$ and a spreading metric LP
relaxation that yields an $O(\log n)$-approximation to the problem. 

\paragraph{On hierarchical clustering more broadly.}
There is an extensive literature on hierarchical clustering and its
applications. It will be impossible to discuss most of it here; for some
applications the reader may refer to
\eg~\citep{jardine1972mathematical,sneath1962numerical,felsenstein2004inferring,
castro2004likelihood}.
Algorithms for hierarchical clustering have received a lot of attention from a
practical perspective. For a definition and overview of 
\emph{agglomerative} algorithms (such as average-linkage, complete-linkage, and
single-linkage) see \eg~\citep{friedman2001elements} and for \emph{divisive
algorithms} see \eg~\cite{Steinbach00acomparison}.

Most previous theoretical work on hierarchical clustering aimed at
evaluating the cluster tree output by the linkage algorithms using the
traditional objective functions for partition-based clustering, \eg
considering $k$-median or $k$-means cost of the clusters induced by the
top levels of the tree (see \eg \citep{Plaxton03,DL05,LNRW06}).  Previous
work also proved that average-linkage can be useful to recover an
underlying partition-based clustering when it exists under certain
stability conditions (see~\citep{BalcanBV10,BaL16}). The approach of this
paper is different: we aim at associating a cost or a value to each
hierarchical clustering and finding the best hierarchical clustering with
respect to these objective functions.

In Section~\ref{S:costfun}, we take an axiomatic approach toward
\emph{objective functions}. Axiomatic approach toward a \emph{qualitative
analysis of algorithms} for clustering where taken before.  For example, the
celebrated result of~\cite{kleinberg2002impossibility} (see
also~\cite{zadeh2009uniqueness}) showed that there is no algorithm satisfying
three natural axioms simultaneously.  This approach was applied to hierarchical
clustering algorithms by~\cite{CM:2010} who showed that in the case of
hierarchical clustering one gets a positive result, unlike the impossibility
result of Kleinberg. Their focus was on finding an ultrametric (on the
datapoints) that is the closest to the metric (in which the data lies) in terms
of the Gromov-Hausdorf distance.  Our approach is completely different as we
focus on defining objective functions and use these for \emph{quantitative
analyses of algorithms}.

Our condition for inputs to have a ground-truth cluster tree, and especially
their $\delta$-adversarially perturbed versions, can be to be in the same
spirit as that of the stability condition of~\citet{BiL12} or \citet{BDLS13}:
the input induces a natural clustering to be recovered whose cost is optimal.
It bears some similarities with the ``strict separation'' condition
of~\citet{BalcanBV10}, while we do not require the separation to be strict, we
do require some additional hierarchical constraints. There are a variety of
stability conditions that aim at capturing some of the structure that
real-world inputs may exhibit (see \eg~\citep{ABS12,BBG13,BalcanBV10,ORSS12}).
Some of them induce a condition under which an underlying clustering can be
mostly recovered (see \eg~\citep{BiL12,BBG09,BBG13} for deterministic
conditions and \eg~\cite{ArK01,BrV08,DaS07,Das99,BRT09} for probabilistic
conditions).  Imposing other conditions allows one to bypass
hardness-of-approximation results for classical clustering objectives (such as
$k$-means), and design efficient approximation algorithms (see,
\eg~\citep{ABS10,AwS12,KuK10}).  \cite{EBW16} also investigate the question of
understanding hierarchical cluster trees for random graphs generated from
graphons. Their goal is quite different from ours---they consider the
``single-linkage tree'' obtained using the graphon as the ground-truth tree and
investigate how a cluster tree that has low \emph{merge distortion} with
respect to this be obtained.%
\footnote{This is a simplistic characterization of their work. However, a more
precise characterization would require introducing a lot of terminology from
their paper, which is not required in this paper.}  
This is quite different from the approach taken in our work which is primarily
focused on understanding performance with respect to admssible cost functions. 


\section{Preliminaries}
\label{sec:prelim}
\subsection{Notation}

An undirected weighted graph $G = (V, E, w)$ is defined by a finite set of
vertices $V$, a set of edges $E \subseteq \{\{u, v\} ~|~ u, v \in V \}$ and a
weight function $w : E \rightarrow \reals_+$, where $\reals_+$ denotes
non-negative real numbers. We will only consider graphs with positive weights
in this paper. To simplify notation (and since the graphs 
are undirected) we let 
$w(u,v) = w(v,u) = w(\{u,v\})$.
 When the weights on the edges are not pertinent, we simply
denote graphs as $G = (V,E)$. When $G$ is clear from the context,
we denote $|V|$ by $n$ and $|E|$ by $m$.
We define $G[U]$ to be the subgraph induced by the nodes of $U$.

A \emph{cluster tree} or \emph{hierarchical clustering} 
$T$ for graph $G$ is a rooted
binary tree with exactly $|V|$ leaves, each of which is labeled by a distinct
vertex $v \in V$.\footnote{In general, one can look at trees that are not
binary.  However, it is common practice to use binary trees in the context of
hierarchical trees.  Also, for results presented in this paper nothing is
gained by considering trees that are not binary.} Given a graph $G = (V, E)$
and a cluster tree $T$ for $G$, for nodes $u, v \in V$ we denote by
$\lca_{T}(u, v)$ the lowest common ancestor (furthest from the root) of $u$ and
$v$ in $T$.

For any internal node $N$ of $T$, we denote the subtree of $T$ rooted at $N$ by
$T_N$.\footnote{For any tree $T$, when we refer to a subtree $T'$ (of $T$)
rooted at a node $N$, we mean the connected subgraph containing all the leaves
of $T$ that are descendant of $N$.}  Moreover, for any node $N$ of $T$, define
$V(N)$ to be the set of leaves of the subtree rooted at $N$.  Additionally, for
any two trees $T_1,T_2$, define the \emph{union} of $T_1,T_2$ to be the tree
whose root has two children $C_1,C_2$ such that the subtree rooted at $C_1$ is
$T_1$ and the subtree rooted at $C_2$ is $T_2$.

Finally, given a weighted graph $G=(V,E,w)$, for any set of vertices $A
\subseteq V$, let $w(A) = \sum_{a,b \in A} w(a,b)$ and for any set of edges
$E_0$, let $w(E_0) = \sum_{e \in E_0} w(e)$. Finally, for any sets of vertices
$A,B \subseteq V$, let $w(A,B) = \sum_{a \in A,b \in B} w(a,b)$.

\subsection{Ultrametrics}
\label{SubSec:ultrametric}
\begin{definition}[Ultrametric] \label{defn:ultrametric} 
  A metric space $(X, d)$ is an ultrametric if for every $x, y, z \in X$,
  $d(x, y) \leq \max \{ d(x, z), d(y, z) \}$.  
\end{definition}

\subsubsection*{Similarity Graphs Generated from Ultrametrics}

We say that a weighted graph $G = (V, E, w)$ is a \emph{similarity graph
generated from an ultrametric}, if there exists an ultrametric $(X, d)$, such
that $V \subseteq X$, and for every $x, y \in V, x \neq y$, $e = \{x, y\}$
exists, and $w(e) = f(d(x, y))$, where $f : \reals_+ \rightarrow \reals_+$ is a
non-increasing function.\footnote{In some cases, we will say that $e = \{x, y\}
\not\in E$, if $w(e) = 0$. This is fine as long as $f(d(x, y)) = 0$.}

\subsubsection*{Dissimilarity Graphs Generated from Ultrametrics}

We say that a weighted graph $G = (V, E, w)$ is a \emph{dissimilarity graph
generated from an ultrametric}, if there exists an ultrametric $(X, d)$, such
that $V \subseteq X$, and for every $x, y \in V, x \neq y$, $e = \{x, y\}$
exists, and $w(e) = f(d(x, y))$, where $f : \reals_+ \rightarrow \reals_+$ is a
non-decreasing function.


\subsubsection*{Minimal Generating Ultrametric}

For a weighted undirected graph $G = (V, E, w)$ generated from an ultrametric
(either similarity or dissimilarity), in general there may be several
ultrametrics and the corresponding function $f$ mapping distances in the
ultrametric to weights on the edges, that generate the same graph. It is useful
to introduce the notion of a minimal ultrametric that generates $G$. 

We focus on similarity graphs here; the notion of minimal 
generating ultrametric for dissimilarity
graphs is easily obtained by suitable modifications. Let $(X, d)$ be an
ultrametric that generates $G = (V, E, w)$ and $f$ the corresponding function
mapping distances to similarities. Then we consider the ultrametric $(V,
\tilde{d})$ defined as follows: (i) $\tilde{d}(u, u) = 0$ and (ii) for $u \neq
v$, 
\begin{align}
	\tilde{d}(u, v) = \tilde{d}(v, u) = \max_{u^\prime, v^\prime} \{ d(u^\prime, v^\prime) ~|~ f(d(u^\prime, v^\prime)) = f(d(u, v)) \} \label{eqn:minUM}
\end{align}
It remains to be seen that $(V, \tilde{d})$ is indeed an ultrametric. First,
notice that by definition, $\tilde{d}(u, v) \geq d(u, v)$ and hence clearly
$\tilde{d}(u, v) = 0$ if and only if $u = v$ as $d$ is the distance in an
ultrametric. The fact that $\tilde{d}$ is symmetric is immediate from the
definition. The only part remaining to check is the so called \emph{isosceles
triangles with longer equal sides} conditions---the ultrametric requirement
that for any $u, v, w$, $d(u, v) \leq \max\{d(u, w), d(v, w) \}$ implies that
all triangles are isosceles and the two sides that are equal are 
at least as large as the
third side. Let $u, v, w \in V$, and assume without loss of generality that
according to the distance $d$ of $(V, d)$, $d(u, w) = d(v, w) \geq d(u, v)$.
From~\eqref{eqn:minUM} it is clear that $\tilde{d}(u, w) =  \tilde{d}(v, w)
\geq d(u, w)$. Also, from~\eqref{eqn:minUM} and the non-increasing nature of
$f$ it is clear that if $d(u, v) \leq d(u^\prime, v^\prime)$, then
$\tilde{d}(u, v) \leq \tilde{d}(u^\prime, v^\prime)$. Thence, $(V, \tilde{d})$
is an ultrametric. The advantage of considering the minimal ultrametric is the
following: if $\calD= \{ \tilde{d}(u, v) ~|~ u, v \in V, u \neq v \}$ and
$\calW = \{ w(u, v) ~|~ u, v \in V, u \neq v \}$, then the restriction of $f$
from  $\calD \rightarrow \calW$ is actually a bijection. This allows the notion
of a generating tree to be defined in terms of distances in the ultrametric or
weights, without any ambiguity. Applying an analogous definition and 
reasoning yields a similar notion for the dissimilarity case.

\begin{definition}[Generating Tree]  \label{defn:generating-tree}
	Let $G = (V, E, w)$ be a graph generated by a minimal ultrametric $(V, d)$
	(either a similarity or dissimilarity graph). Let $T$ be a rooted binary
	tree with $|V|$ leaves and $|V| - 1$ internal nodes; let $\calN$ denote the
	internal nodes and $L$ the set of leaves of $T$ and let $\sigma : L
	\rightarrow V$ denote a bijection between the leaves of $T$ and nodes of
	$V$. We say that $T$ is a generating tree for $G$, if there exists a weight
	function $W : \calN \rightarrow \reals_+$, such that for $N_1, N_2 \in \calN$,
	if $N_1$ appears on the path from $N_2$ to the root, $W(N_1) \leq W(N_2)$.
	Moreover for every $x, y \in V$, $w(\{x, y\}) = W(LCA_{T}(\sigma^{-1}(x),
	\sigma^{-1}(y)))$. 
\end{definition}


The notion of a generating tree defined above more or less corresponds to what
is referred to as a \emph{dendogram} in the machine learning literature (see
\eg \citep{CM:2010}). More formally, a dendogram is a rooted tree (not
necessarily binary), where the leaves represent the datapoints. Every internal
node in the tree has associated with it a height function $h$ which is the
distance between any pairs of datapoints for which it is the least common
ancestor. It is a well-known fact that a set of points in an ultrametric can be
represented using a dendogram (see \eg \citep{CM:2010}).  A dendogram can
easily be modified to obtain a generating tree in the sense of
Definition~\ref{defn:generating-tree}: an internal node with $k$ children
is
replace by an arbitrary binary tree with $k$ leaves and 
the children of the
nodes in the dendogram are attached to these $k$ leaves. 
The height $h$ of this
node is used to give the weight $W = f(h)$ to all the $k-1$ 
internal nodes
added when replacing this node. Figure~\ref{fig:dendogram} shows this
transformation.

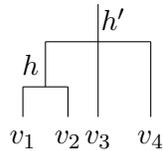
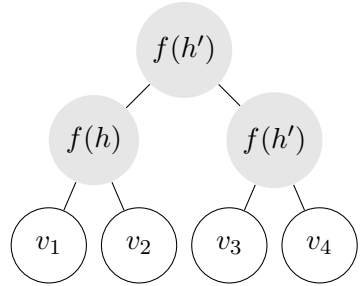
\begin{figure*}[t!]
	\centering
	\begin{subfigure}[t]{.49\textwidth}
		\centering
		\begin{tikzpicture}
			\draw (0, -0.5) -- (0, -1);
			\draw (-0.7, -1) -- (0.7, -1);
			\draw (-0.7, -1) -- (-0.7, -1.6);
			\draw (0, -1) -- (0, -2.0);
			\draw (0.7, -1) -- (0.7, -2.0);
			\draw (-1, -1.6) -- (-0.4, -1.6);
			\draw (-1, -1.6) -- (-1, -2);
			\draw (-0.4, -1.6) -- (-0.4, -2);
			\node at (-1, -2.3) {$v_1$};
			\node at (-0.4, -2.3) {$v_2$};
			\node at (0, -2.3) {$v_3$};
			\node at (0.7, -2.3) {$v_4$};
			\node at (-0.9, -1.3) {$h$};
			\node at (0.2, -0.7) {$h^\prime$};
		\end{tikzpicture}
		\caption{Dendogram on $4$ nodes. \label{fig:dend}}
	\end{subfigure}
	\begin{subfigure}[t]{.45\textwidth}
		\centering
		\begin{tikzpicture}
			\node[fill=gray!20, circle, minimum width=1.2cm] (root) at (0, 0) {$f(h^\prime)$};
			\node[fill=gray!20, circle, minimum width=1.2cm] (l1) at (-1.2, -1.2) {$f(h)$};
			\node[fill=gray!20, circle, minimum width=1.2cm] (r1) at (1.2, -1.2) {$f(h^\prime)$};
			\node[draw, circle, minimum width=1.0cm] (v1) at (-1.8, -2.6) {$v_1$};
			\node[draw, circle, minimum width=1.0cm] (v2) at (-0.6, -2.6) {$v_2$};
			\node[draw, circle, minimum width=1.0cm] (v3) at (0.6, -2.6) {$v_3$};
			\node[draw, circle, minimum width=1.0cm] (v4) at (1.8, -2.6) {$v_4$};

			\draw (root) -- (l1);
			\draw (root) -- (r1);
			\draw (l1) -- (v1);
			\draw (l1) -- (v2);
			\draw (r1) -- (v3);
			\draw (r1) -- (v4);

		\end{tikzpicture}
		\caption{Generating tree equivalent to dendogram\label{fig:gentree}}
	\end{subfigure}
	\caption{\label{fig:dendogram}Dendogram and equivalent generating tree.}
\end{figure*}

\subsubsection*{Ground-Truth Inputs}

\begin{definition}[Ground-Truth Input.]
  We say that a graph $G$ is a \emph{\structinput} if it is a similarity or
	dissimilarity graph generated from an ultrametric. Equivalently, there
	exists a tree $T$ that is generating for $G$.
\end{definition}

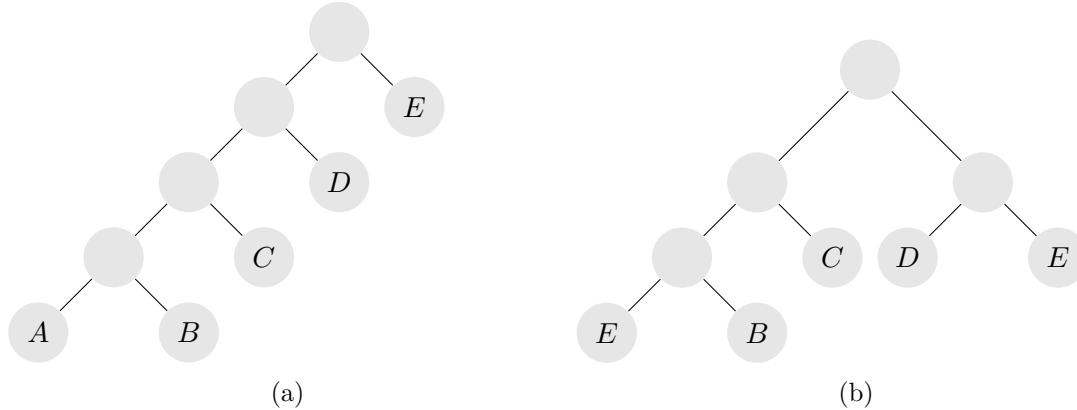
\begin{figure*}[t!]
	\centering
	\begin{subfigure}[t]{0.45\textwidth}
	\begin{tikzpicture}
		\node[fill=gray!20, circle, minimum width=0.8cm] (root) at (0, 0) {};
		\node[fill=gray!20, circle, minimum width=0.8cm] (l1) at (-1, -1) {};
		\node[fill=gray!20, circle, minimum width=0.8cm] (l2) at (-2, -2) {};
		\node[fill=gray!20, circle, minimum width=0.8cm] (l3) at (-3, -3) {};
		\node[fill=gray!20, circle, minimum width=0.8cm] (v5) at (1, -1) {$E$};
		\node[fill=gray!20, circle, minimum width=0.8cm] (v4) at (0, -2) {$D$};
		\node[fill=gray!20, circle, minimum width=0.8cm] (v3) at (-1, -3) {$C$};
		\node[fill=gray!20, circle, minimum width=0.8cm] (v2) at (-2, -4) {$B$};
		\node[fill=gray!20, circle, minimum width=0.8cm] (v1) at (-4, -4) {$A$};
		\draw (root) -- (l1);
		\draw (root) -- (v5);
		\draw (l1) -- (l2);
		\draw (l1) -- (v4);
		\draw (l2) -- (l3);
		\draw (l2) -- (v3);
		\draw (l3) -- (v2);
		\draw (l3) -- (v1);
	\end{tikzpicture}
		\caption{ }
	\end{subfigure}
	\begin{subfigure}[t]{0.45\textwidth}
	\begin{tikzpicture}
		\node[fill=gray!20, circle, minimum width=0.8cm] (root) at (0, 0) {};
		\node[fill=gray!20, circle, minimum width=0.8cm] (l1) at (-1.5, -1.5) {};
		\node[fill=gray!20, circle, minimum width=0.8cm] (l2) at (-2.5, -2.5) {};
		\node[fill=gray!20, circle, minimum width=0.8cm] (r1) at (1.5, -1.5) {};
		\node[fill=gray!20, circle, minimum width=0.8cm] (v5) at (2.5, -2.5) {$E$};
		\node[fill=gray!20, circle, minimum width=0.8cm] (v4) at (0.5, -2.5) {$D$};
		\node[fill=gray!20, circle, minimum width=0.8cm] (v3) at (-0.5, -2.5) {$C$};
		\node[fill=gray!20, circle, minimum width=0.8cm] (v2) at (-1.5, -3.5) {$B$};
		\node[fill=gray!20, circle, minimum width=0.8cm] (v1) at (-3.5, -3.5) {$E$};
		\draw (root) -- (l1);
		\draw (root) -- (r1);
		\draw (r1) -- (v4);
		\draw (r1) -- (v5);
		\draw (l1) -- (l2);
		\draw (l1) -- (v3);
		\draw (l2) -- (v2);
		\draw (l2) -- (v1);
	\end{tikzpicture}
		\caption{ }
	\end{subfigure}
	\caption{(a) Caterpillar tree on $5$ nodes with unit-weight edges used to define a
	tree metric. (b) A candidate cluster tree for the data generated using the
	tree metric \label{fig:cater1}}
\end{figure*}

\noindent\textbf{Motivation}. We briefly describe the motivation for defining
graphs generated from an ultrametric as ground-truth inputs. We'll focus the
discussion on similarity graphs, though essentially the same logic holds for
dissimilarity graphs. As described earlier, there is a natural notion of a
\emph{generating tree} associated with graphs generated from ultrametrics. This
tree itself can be viewed as a cluster tree. The clusters obtained using the
generating tree have the property that any two nodes in the same cluster are at
least as similar to each other as they are to points outside this cluster; and
this holds at every level of granularity. Furthermore, as observed
by~\citet{CM:2010}, many practical hierarchical clustering algorithms such as
the linkage based algorithms, actually output a dendogram equipped with a
height function, that corresponds to an ultrametric embedding of the data.
While their work focuses on algorithms that find embeddings in ultrametrics,
our work focuses on finding cluster trees. We remark that these problems are
related but also quite different.

Furthermore, our results show that the linkage algorithms (and some other
practical algorithms), recover a generating tree when given as input graphs
that are generated from an ultrametric. Finally, we remark that relaxing the
notion further leads to instances where it is hard to define a `natural'
ground-truth tree.  Consider a similarity graph generated by a
\emph{tree-metric} rather than an ultrametric, where the tree is the
caterpillar graph on 5 nodes (see Fig.~\ref{fig:cater1}(a)). Then, it is hard
to argue that the tree shown in Fig.~\ref{fig:cater1}(b) is not a more suitable
cluster tree. For instance, $D$ and $E$ are more similar to each other than $D$
is to $B$ or $A$. In fact, it is not hard to show that by choosing a suitable
function $f$ mapping distances from this tree metric to similarities,
Dasgupta's objective function is minimized by the tree shown in
Fig.~\ref{fig:cater1}(b), rather than the `generating' tree in
Fig.~\ref{fig:cater1}(a).


\section{Quantifying Output Value: An Axiomatic Approach}
\label{S:costfun}
\label{sec:cost-functions}
\subsection{Admissible Cost Functions}

Let us focus on the similarity case; in this case we use \emph{cost} and
\emph{objective} interchangeably.  Let $G = (V, E, w)$ be an undirected
weighted graph and let $T$ be a cluster tree for graph $G$. We want to consider
cost functions for cluster trees that capture the quality of the hierarchical
clustering produced by $T$. Following the recent work of~\citet{Das:2016}, we
adopt an approach in which a cost is assigned to each internal node of the tree
$T$ that corresponds to the quality of the split at that node. \smallskip

\noindent\textbf{The Axiom}. A natural property we would like the cost function to
satisfy is that a cluster tree $T$ has minimum cost if and only if $T$ is a
generating tree for $G$.  Indeed, the objective function can then be used to
indicate whether a given tree is generating and so, whether it is an underlying
ground-truth hierarchical clustering.  Hence, the objective function acts as a
``guide'' for finding the correct hierarchical classification.  Note that there
may be multiple trees that are generating for the same graph.  For example, if
$G = (V, E, w)$ is a clique with every edge having the same weight then every
tree is a generating tree.  In these cases, all the generating tree are
\emph{valid} ground-truth hierarchical clusterings. 

Following~\cite{Das:2016}, we restrict the search space for such cost
functions. For an internal node $N$ in a clustering tree $T$, let $A, B
\subseteq V$ be the leaves of the subtrees rooted at the left and right child
of $N$ respectively.  We define the cost $\Costtree$ of the tree $T$ as the sum of the cost at every internal node $N$
in the tree,  and at an individual node $N$ we consider cost functions $\costnode$ of
the form 
\begin{align}
	\Costtree(T) &= \sum_{N} \costnode(N),  \label{eqn:costfuncfull}
\end{align}
\begin{align}
	\costnode(N) &= \left( \sum_{x \in A, y \in B} w(x, y) \right) \cdot g( |A|, |B|)
	\label{eqn:costfunc}
\end{align}
We remark that~\cite{Das:2016} defined $g(a,b) = a+b$.

\begin{definition}[Admissible Cost Function] \label{defn:admissibleCF}
  We say that a cost function $\gamma$ of the 
  form~(\ref{eqn:costfuncfull},\ref{eqn:costfunc}) is
  \emph{admissible} if it satisfies the condition that for all 
  similarity graphs $G =(V, E, w)$
  generated from a minimal ultrametric $(V, d)$, a cluster tree $T$ for $G$
  achieves the minimum cost if and only if it is a generating tree for $G$.
\end{definition}

\begin{remark}
  \label{rem:costfun:diss}
  Analogously, for the dissimilarity setting we define \emph{admissible
  value functions} to be the functions of the form~(\ref{eqn:costfuncfull},\ref{eqn:costfunc})
  that satisfy: for all dissimilarity graph $G$ generated 
  from a minimal ultrametric $(V, d)$, a cluster tree $T$ for $G$
  achieves the maximum value if and only if it is a generating 
  tree for $G$.
\end{remark}

\begin{remark}
	The RHS of~\eqref{eqn:costfunc} has linear dependence on the weight of the
	cut $(A, B)$ in the subgraph of $G$ induced by the vertex set $A \cup B$ as
	well as on an arbitrary function of the number of leaves in the subtrees of
	the left and right child of the internal node creating the cut $(A, B)$.
	For the purpose of hierarchical clustering this form is fairly natural and
	indeed includes the specific cost function introduced by~\citet{Das:2016}.
	We could define the notion of admissibility for other forms of the cost
	function similarly and  it would be of interest to understand whether they
	have properties that are desirable from the point of view of hierarchical
	clustering.
\end{remark}

\subsection{Characterizing Admissible Cost Functions}

In this section, we give an almost complete characterization of admissible cost
functions of the form~\eqref{eqn:costfunc}. The following theorem shows
that cost functions of this form are admissible if and only if they satisfy
three conditions: that all cliques must have the same cost, symmetry and
monotonicity.

\begin{theorem} \label{thm:cost-func}
	Let $\costnode$ be a cost function of the form~\eqref{eqn:costfunc} and let
	$g$ be the corresponding function used to define $\costnode$. Then
	$\costnode$ is admissible if and only if it satisfies the following three
	conditions.
	\begin{enumerate}
		\item Let $G = (V, E, w)$  be a clique, \ie for every $x, y \in V$, $e =
			\{x, y \} \in E$ and $w(e) = 1$ for every $e \in E$. Then the cost
			$\Costtree(T)$ for every cluster tree $T$ of $G$ is identical.
		\item For every $n_1, n_2 \in \naturals$, $g(n_1, n_2) = g(n_2, n_1)$.
		\item For every $n_1, n_2 \in \naturals$, $g(n_1 + 1, n_2) > g(n_1, n_2)$.
	\end{enumerate}
\end{theorem}
\begin{proof} We first prove the only if part and then the if part. \smallskip \\
	\noindent\underline{Only If Part}: Suppose that $\costnode$ is indeed an
	admissible cost function. We prove that all three conditions must be
	satisfied by $\costnode$. \smallskip \\
	\noindent \emph{1. All cliques have same cost.} We observe that a clique $G
	= (V, E, w)$ can be generated from an ultrametric. Indeed, let $X = V$ and
	let $d (u, v) = d(v, u) = 1$ for every $u, v\in X$ such that $u \neq v$ and
	$d(u, u) = 0$. Clearly, for $f : \reals^+ \rightarrow \reals^+$ that is
	non-increasing and satisfying $f(1) = 1$, $(V, d)$ is a minimal ultrametric
	generating $G$. 

	Let $T$ be any binary rooted tree with leaves labeled by $V$, \ie a cluster
	tree for graph $G$. For any internal node $N$ of $T$ define $W(N) = 1$ as
	the weight function. This satisfies the definition of generating tree
	(Defn.~\ref{defn:generating-tree}). Thus, every cluster tree $T$ for $G$ is
	generating and hence, by the definition of admissibility all of them must be
	optimal, \ie they all must have exactly the same cost. \smallskip \\
	\noindent \emph{2. $g(n_1, n_2) = g(n_2, n_1)$.} This part follows more or
	less directly from the previous part. Let $G$ be a clique on $n_1 + n_2$
	nodes. Let $T$ be any cluster tree for $G$, with subtrees $T_1$ and $T_2$
	rooted at the left and right child of the root respectively, such that $T_1$
	contains $n_1$ leaves and $T_2$ contains $n_2$ leaves. The number of edges,
	and hence the total weight of the edges,  crossing the cut induced by the
	root node of $T$ is $n_1 \cdot n_2$.  Let $\tilde{T}$ be a tree obtained by
	making $T_2$ be rooted at the left child of the root and $T_1$ at the right
	child. Clearly $\tilde{T}$ is also a cluster tree for $G$ and induces the
	same cut at the root node, hence using the property that all cliques have
	the same cost,  $\Costtree(T) = \Costtree(\tilde{T})$. But $\Costtree(T) =
	n_1 \cdot n_2 \cdot g(n_1, n_2) + \Costtree(T_1) + \Costtree(T_2)$ and
	$\Costtree(\tilde{T}) = n_1 \cdot n_2 \cdot g(n_2, n_1) + \Costtree(T_1) +
	\Costtree(T_2)$. Thence, $g(n_1, n_2) = g(n_2, n_1)$.  \smallskip \\
	\noindent \emph{3. $g(n_1 + 1, n_2) > g(n_1, n_2)$.} Consider a graph on
	$n_1 + n_2 + 1$ nodes generated from an ultrametric as follows. Let $V_1 =
	\{ v_1, \ldots, v_{n_1} \}$, $V_2 = \{ v^\prime_1, \ldots, v^\prime_{n_2}
	\}$ and consider the ultrametric $(V_1 \cup V_2 \cup \{v^*\}, d)$ defined by
	$d(x, y) = 1$ if $x \neq y$ and $x, y \in V_1$ or $x, y \in V_2$, $d(x, y) =
	2$ if $x \neq y$ and $x \in V_1, y \in V_2$ or $x \in V_2, y \in V_1$,
	$d(v^*, x) = d(x, v^*) = 3$ for $x \in V_1 \cup V_2$, and $d(u, u) = 0$ for
	$u \in V_1 \cup V_2 \cup \{v^* \}$. It can be checked easily by enumeration
	that this is indeed an ultrametric. Furthermore, if $f : \reals^+
	\rightarrow \reals^+$ is non-increasing and satisfies $f(1) = 2$, $f(2) = 1$
	and $f(3) = 0$, \ie $w(\{u, v\}) = 2$ if $u$ and $v$ are both either in
	$V_1$ or $V_2$, $w(\{u, v\}) = 1$ if $u \in V_1$ and $v \in V_2$ or the
	other way around, and $w(\{v^*, u \}) = 0$ for $u \in V_1 \cup V_2$, then
	$(V_1 \cup V_2, \{ v^* \}, d)$ is a minimal ultrametric generating $G$.

	Now consider two possible cluster trees defined as follows: Let $T_1$ be an
	arbitrary tree on nodes $V_1$, $T_2$ and arbitrary tree on nodes $V_2$. $T$
	is obtained by first joining $T_1$ and $T_2$ using internal node $N$ and
	making this the left subtree of the root node $\rho$ and the right subtree
	of the root node is just the singleton node $v^*$.  $T^\prime$ is obtained
	by first creating a tree by joining $T_1$ and the singleton node $v^*$ using
	internal node $N^\prime$, this is the left subtree of the root node
	$\rho^\prime$ and $T_2$ is the right subtree of the root node. (See
	Figures~\ref{fig:treeT} and~\ref{fig:treeTp}.)

	\begin{figure*}[t!]
		\centering
		\begin{subfigure}[t]{0.45 \textwidth}
			\centering
			\begin{tikzpicture}
				\node[fill=gray!20, circle, minimum width=0.9cm] (root) at (0, 0) {$\rho$};
				\node[fill=gray!20, circle, minimum width=0.9cm] (lchild) at (-1, -1) {$N$};
				\node[draw, circle, minimum width=0.9cm] (leaf) at (1.5, -1.5) {$v^*$};
				\node (T1) at (-1.7, -3.2) {$T_1$};
				\node (T2) at (-0.3, -3.2) {$T_2$};
				\node[fill=gray!20, circle, minimum width=3pt] (T1head) at (-1.7, -2) {};
				\node[fill=gray!20, circle, minimum width=3pt] (T2head) at (-0.3, -2) {};
				\draw (-1.7, -2.1) -- (-2.2, -3.5) -- (-1.2, -3.5) -- (-1.7, -2.1);
				\draw (-0.3, -2.1) -- (-0.8, -3.5) -- (0.2, -3.5) -- (-0.3, -2.1);
				\draw (lchild) -- (T1head);
				\draw (lchild) -- (T2head);
				\draw (root) -- (leaf);
				\draw (root) -- (lchild);
			\end{tikzpicture} \smallskip

			\caption{Tree $T$ \label{fig:treeT}}
		\end{subfigure}
		\begin{subfigure}[t]{0.45 \textwidth}
			\centering
			\begin{tikzpicture}
				\node[fill=gray!20, circle, minimum width=0.9cm] (root) at (0, 0) {$\rho^\prime$};
				\node[fill=gray!20, circle, minimum width=0.9cm] (lchild) at (-1, -1) {$N^\prime$};
				\node[fill=gray!20, circle, minimum width=3pt] (T1head) at (-1.7, -2) {};
				\draw (-1.7, -2.1) -- (-2.2, -3.5) -- (-1.2, -3.5) -- (-1.7, -2.1);
				\node (T1) at (-1.7, -3.2) {$T_1$};
				\node[draw, circle, minimum width=0.9cm] (leaf) at (-0.3, -2.5) {$v^*$};

				\node[fill=gray!20, circle, minimum width=3pt] (T2head) at (1, -1) {};
				\draw (1, -1.1) -- (0.5, -2.5) -- (1.5, -2.5) -- (1, -1.1);

				\draw (lchild) -- (T1head);
				\draw (lchild) -- (leaf);
				\draw (root) -- (T2head);
				\draw (root) -- (lchild);
				\node (T2) at (1, -2.2) {$T_2$};
			\end{tikzpicture} \smallskip 

			\caption{Tree $T^\prime$ \label{fig:treeTp}}
		\end{subfigure}
		\caption{\label{fig:trees}Trees $T$ and $T^\prime$ used to show monotonicity of $g$.}

	\end{figure*}
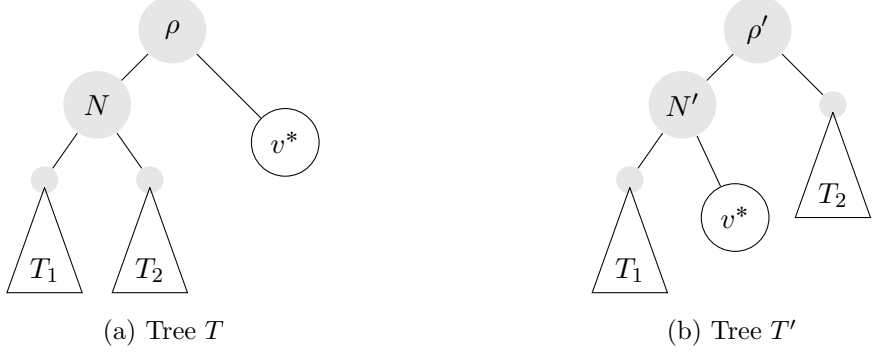

	Now it can be checked that $T$ is generating by defining the following
	weight function. For every internal node $M$ of $T_1$, let $W(M) = 1$,
	similarly for every internal node $M$ of $T_2$, let $W(M) = 1$, define $W(N)
	= 2$ and $W(\rho) = 3$. Now, we claim that $T^\prime$ cannot be a generating
	tree. This follows from the fact that for a node $u \in V_1, v \in V_2$, the
	root node $\rho^\prime =\lca_{T^\prime}(u, v)$, but it is also the case that
	$\rho^\prime =  \lca_{T^\prime}(v^*, v)$. Thus, it cannot possibly be the
	case that $W(\rho) = w(\{u, v\})$ and $W(\rho) = w(\{v^*, v\})$  as $w(\{u,
	v\}) \neq w(\{v^*, v\})$.  By definition of admissibility, it follows that
	$\Costtree(T) < \Costtree(T^\prime)$, but $\Costtree(T) = \Costtree(T_1) +
	\Costtree(T_2) + n_1 \cdot n_2 \cdot g(n_1, n_2)$. The last term arises from
	the cut at node $N$; the root makes no contribution as the cut at the root
	node $\rho$ has weight $0$. On the other hand $\Costtree(T^\prime) =
	\Costtree(T_1) + \Costtree(T_2) + n_1 \cdot n_2 \cdot g(n_1 + 1, n_2)$.
	There is no cost at the node $N^\prime$, since the cut has size $0$;
	however, at the root node the cost is now $n_1 \cdot n_2 \cdot g(n_1 + 1,
	n_2)$ as the left subtree at the root contains $n_1 + 1$ nodes. It follows
	that $g(n_1 + 1, n_2) > g(n_1 ,n_2)$. \smallskip 

	\noindent\underline{If Part}: 
	For the other direction, we first use the following observation. By
	condition 2 in the statement of the theorem, every clique on $n$ nodes
	has the same cost irrespective of the tree used for hierarchical clustering;
	let $\kappa(n)$ denote said cost. Let $n_1, n_2 \geq 1$, then we have, 
	\begin{align}
		n_1 \cdot n_2 \cdot g(n_1, n_2) &= \kappa(n_1 + n_2) - \kappa(n_1) - \kappa(n_2) \label{eqn:defg}
	\end{align}

	We will complete the proof by induction on $|V|$. The minimum number of
	nodes required to have a cluster tree with at least one internal node is
	$2$. Suppose $|V| =  2$, then there is a unique (up to interchanging left
	and right children) cluster tree; this tree is also generating and hence by
	definition any cost function is admissible. Thus, the base case is covered
	rather easily.

	Now, consider a graph $G = (V, E, w)$ with $|V| = n > 2$. Let $T^*$ be a
	tree that is generating. Suppose that $T$ is any other tree. Let $\rho^*$
	and $\rho$ be the root nodes of the trees respectively. Let $V^*_L$ and
	$V^*_R$ be the nodes on the left subtree and right subtree of $\rho^*$;
	similarly $V_L$ and $V_R$ in the case of $\rho$. Let $A = V^*_L \cap V_L$,
	$B = V^*_L \cap V_R$, $C = V^*_R \cap V_L$, $D = V^*_R \cap V_R$. Let $a$,
	$b$, $c$ and $d$ denote the sizes of $A$, $B$, $C$ and $D$ respectively.

	We will consider the case when all of $a, b, c, d > 0$; the proof is similar
	and simpler in case some of them are $0$.  Let $\tilde{T}$ be a tree with
	root $\tilde{\rho}$ that has the following structure: Both children of the
	root are internal nodes, all of $A$ appears as leaves in the left subtree of
	the left child of the root, $B$ as leaves in the right subtree of the left
	child of the root, $C$ as leaves in the left subtree of the right child of
	the root and $D$ as leaves in the right subtree of the right child of the
	root. We assume that all four subtrees for the sets $A$, $B$, $C$, $D$ are
	generating and hence by induction optimal.  We claim that the cost of
	$\tilde{T}$ is at least as much as the cost of $T^*$. To see this note that
	$V^*_L = A \cup B$. Thus, the left subtree of $\rho^*$ is optimal for the
	set $V^*_L$ (by induction), whereas that of $\tilde{\rho}$ may or may not
	be. Similarly for all the nodes in $V^*_R$. The only other thing left to
	account for is the cost at the root. But since $\rho^*$ and $\tilde{\rho}$
	induce exactly the same cut on $V$, the cost at the root is the same. Thus,
	$\Costtree(\tilde{T}) \geq \Costtree(T^*)$. Furthermore, equality holds if
	and only if $\tilde{T}$ is also generating for $G$.

	Let $W^*$ denote the weight function for the generating tree $T^*$ such that
	for all $u, v \in V$, $W^*(\lca_{T^*}(u, v)) = w(\{u, v\})$. Let $\rho^*_L$ and
	$\rho^*_R$ denote the left and right children of the root $\rho^*$ of $T^*$.
	For all $u_a \in A, u_b \in B$, $w(\{u_a, u_b\}) \geq W^*(\rho^*_L)$. Let
	\[ x = \frac{1}{a b} \sum_{u_a \in A, u_b \in B} w(\{u_a, u_b \}) 
	\]
	denote the average weight of the edges going between $A$ and $B$; it follows that $x \geq
	W^*(\rho^*_L)$. Similarly for all $u_c \in C, u_d \in D$, $w(\{u_c, u_d\}) \geq
	W^*(\rho^*_R)$. Let 
	\[ y = \frac{1}{c d} \sum_{u_c \in C, u_d \in D} w(\{u_c, u_d \}) 
	\]
	denote the average weight of the edges going between $C$ and $D$; it follows
	that $y \geq W^*(\rho^*_R)$. Finally for every $u \in A \cup B, u^\prime \in
	C \cup D$, $w(\{u, u^\prime\}) = W^*(\rho^*)$; denote this common value by
	$z$. By the definition of generating tree, we know that $x \geq z$ and $y
	\geq z$.

	Now consider the tree $T$. Let $T_L$ and $T_R$ denote the left and right
	subtrees of $\rho$. By induction, it must be that $T_L$ splits $A$ and $C$
	as the first cut (or at least that's one possible tree, if multiple cuts
	exist), similarly $T_R$ first cuts $B$ and $D$. Both, $T$ and $\tilde{T}$
	have subtrees containing only nodes from $A$, $B$, $C$ and $D$. The costs
	for these subtrees are identical in both cases (by induction). Thus, we have
	\begin{align*}
		\Costtree(T) - \Costtree(\tilde{T}) &= z ac \cdot g(a, c) + z bd  \cdot g(b, d) +  (x ab + y cd + z ( ad + bc))\cdot g(a + c, b + d) \\
		&~~~~ - x ab \cdot g(a, b) + y \cdot cd g (c, d) - z (a + b) (c + d) \cdot g(a + b, c + d) \\
		&= (x - z) ab (g(a + c, b +d ) - g(a, b)) + (y - z) cd (g(a + c, b + d) - g(c, d)) \\
		&~~~~ + z ((a + c) (b + d) \cdot g(a + c, b + d) + ac \cdot g(a, c) + bd \cdot g(b, d)) \\
		&~~~~ - z((a + b)(c + d) \cdot g(a + b, c + d) + ab \cdot g(a, b) + cd \cdot g(c, d))
		\intertext{Using~\eqref{eqn:defg}, we get that the last two expressions above
		both evaluate to $z (\kappa(a + b + c + d) - \kappa(a) - \kappa(b) -
		\kappa(c) - \kappa(d))$, but have opposite signs. Thus, we get}
		\Costtree(T) - \Costtree(\tilde{T}) &= (x - z) ab (g(a + c, b +d ) - g(a,
		b)) + (y - z) cd (g(a + c, b + d) - g(c, d)) 
	\end{align*}
	It is clear that the above expression is always non-negative and is $0$ if
	and only if $x = z$ and $y = z$. If it is the latter case and it is also the
	case that $\Costtree(\tilde{T}) = \Costtree(T^*)$, then it must actually be
	the case that $T$ is a generating tree.
\end{proof}

\subsubsection{Characterizing $g$ that satisfy conditions of Theorem~\ref{thm:cost-func}}

Theorem~\ref{thm:cost-func} give necessary and sufficient conditions on
$g$ for cost functions of the form~\eqref{eqn:costfunc} be admissible. However,
it leaves open the question of the existence of functions satisfying the
criteria and also characterizing the functions $g$ themselves. The fact that
such functions exist already follows from the work of~\citet{Das:2016}, who
showed that if $g(n_1, n_2) = n_1 + n_2$, then all cliques have the same cost.
Clearly, $g$ is monotone and symmetric and thus satisfies the condition of
Theorem~\ref{thm:cost-func}.

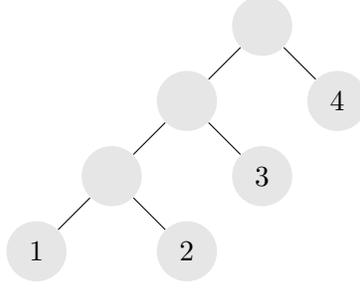
\begin{figure}
	\centering
	\begin{tikzpicture}
		\node[fill=gray!20, circle, minimum width=0.8cm] (root) at (0, 0) {};
		\node[fill=gray!20, circle, minimum width=0.8cm] (l1) at (-1, -1) {};
		\node[fill=gray!20, circle, minimum width=0.8cm] (l2) at (-2, -2) {};
		\node[fill=gray!20, circle, minimum width=0.8cm] (v4) at (1, -1) {$4$};
		\node[fill=gray!20, circle, minimum width=0.8cm] (v3) at (0, -2) {$3$};
		\node[fill=gray!20, circle, minimum width=0.8cm] (v2) at (-1, -3) {$2$};
		\node[fill=gray!20, circle, minimum width=0.8cm] (v1) at (-3, -3) {$1$};
		\draw (root) -- (l1);
		\draw (root) -- (v4);
		\draw (l1) -- (l2);
		\draw (l1) -- (v3);
		\draw (l2) -- (v2);
		\draw (l2) -- (v1);
	\end{tikzpicture}
	\caption{\label{fig:caterpillar} The caterpillar cluster tree for a clique with $4$ nodes.}
\end{figure}

In order to give a more complete characterization, we define $g$ as follows:
Suppose $g(\cdot, \cdot)$ is symmetric, we define $g(n, 1)$ for all $n \geq 1$
so that $g(n, 1)/(n + 1)$ is non-decreasing.\footnote{The function proposed
by~\citet{Das:2016} is $g(n, 1) = n + 1$, so this ratio is always $1$.}  We
consider a particular cluster tree for a clique that is defined using a
caterpillar graph, \ie a cluster tree where the right child of any internal
node is a leaf labeled by one of the nodes of $G$ and the left child is
another internal node, except at the very bottom. Figure~\ref{fig:caterpillar}
shows a caterpillar cluster tree for a clique on $4$ nodes.  The cost of the
clique on $n$ nodes, say $\kappa(n)$, using this cluster tree is given by 
\[ 
	\kappa(n) = \sum_{i = 0}^{n - 1} i \cdot g( i, 1) 
\]
Now, we enforce the condition that all cliques have the same cost by defining
$g(n_1, n_2)$ for $n_1, n_2 > 1$ suitably, in particular, 
\begin{align}\label{eq:honig}
	g(n_1, n_2) = \frac{\kappa(n_1 + n_2) - \kappa(n_1) - \kappa(n_2)}{n_1 \cdot n_2}
\end{align}
Thus it only remains to be shown that $g$ is strictly increasing. We show that
for $n_2 \leq n_1$, $g(n_1 + 1, n_2) > g(n_1, n_2)$. In order to show this it
suffices to show that,
\begin{align*}
	n_1 (\kappa(n_1 + n_2 + 1) - \kappa(n_1 + 1) - \kappa(n_2)) - (n_1 + 1) (\kappa(n_1 + n_2) - \kappa(n_1) - \kappa(n_2)) &> 0
\end{align*}

Thus, consider
\begin{align*}
	&n_1 (\kappa(n_1 + n_2 + 1) - \kappa(n_1 + 1) - \kappa(n_2)) - (n_1 + 1) (\kappa(n_1 + n_2) - \kappa(n_1) - \kappa(n_2))  \\
	&~~=n_1(\kappa(n_1 + n_2 + 1) - \kappa(n_1 + n_2) -  \kappa(1) - \kappa(n_1 + 1) + \kappa(n_1) + \kappa(1)) - (\kappa(n_1 + n_2) - \kappa(n_1)  - \kappa(n_2)) \\
	&~~=n_1  (n_1 + n_2) g(n_1 + n_2, 1) - n_1^2 g(n_1 , 1) - (\kappa(n_1 + n_2) - \kappa(n_1) - \kappa(n_2)) \\
	&~~\geq n_1 (n_1 + n_2) g(n_1 + n_2, 1) - n_1^2 g(n_1, 1) - \sum_{i =
	n_1}^{n_1 + n_2 - 1} i \cdot g(i, 1) \\
	&~~~\geq \frac{g(n_1 + n_2, 1)}{n_1 + n_2 +1 } \cdot \left(n_1 (n_1 + n_2) (n_1 + n_2 + 1) - n_1^2 (n_1 + 1) - \sum_{i = n_1}^{n_1 + n_2 -1} i (i + 1)\right) > 0\\
\end{align*}

Above we used the fact that $g(n, 1)/(n + 1)$ is non-decreasing 
in $n$ and some
elementary calculations. This shows that the objective function proposed
by~\citet{Das:2016} is by no means unique. Only in the last step, 
do we get an
inequality where we use the condition that $g(n, 1)/(n + 1)$ 
is increasing.
Whether this requirement can be relaxed further is 
also an interesting direction.

\subsubsection{Characterizing Objective Functions for Dissimilarity Graphs}

When the weights of the edges represent dissimilarities instead of
similarities, one can consider objective functions of the same form
as~\eqref{eqn:costfunc}. 
As mentioned in Remark~\ref{rem:costfun:diss}, 
the difference in this case is that the goal is to
maximize the objective function and hence the definition of admissibility now
requires that generating trees have a value of the objective that is strictly larger than any
tree that is not generating.

The characterization of admissible objective functions as given in
Theorem~\ref{thm:cost-func} for the similarity case continues to hold in the
case of dissimilarities. The proof follows in the same manner by appropriately
switching the direction of the inequalities when required. 



\section{Similarity-Based Inputs: Approximation Algorithms}
\label{S:sim:worstcase}
In this section, we analyze 
the recursive $\phi$-sparsest-cut
algorithm (see Algorithm~\ref{alg:phiSparsestCut}) 
that was described previously in~\cite{Das:2016}.
For clarity, we work with the cost function introduced 
by~\cite{Das:2016}:
The goal is to find a tree $T$ minimizing $\cost(T) = \sum_{N \in T} \cost(N)$
where for each node $N$ of $T$ with children $N_1$, $N_2$, 
$\cost(N) = w(V(N_1),V(N_2)) \cdot V(N)$.
We show that the $\phi$-sparsest-cut algorithm achieves
a $6.75\phi$-approximation.
(\citet{CC16} also proved an $O(\phi)$ approximation for Dasgupta's function.)
Our proof also yields an approximation guarantee not just for Dasgupta's cost function 
but more generally for any admissible cost function,
but the approximation ratio depends 
on the cost function. 

%
The $\phi$-sparsest-cut algorithm (Algorithm~\ref{alg:phiSparsestCut})
 constructs a binary tree 
 top-down  by 
recursively finding cuts using a $\phi$-approximate sparsest cut algorithm, where
the \emph{sparsest-cut} problem 
asks for a set $A$ minimizing the
\emph{sparsity} $w(A,V\setminus A)/(|A||V\setminus A|)$ of the cut $(A, V\setminus A)$.
\begin{algorithm}
  \caption{Recursive $\phi$-Sparsest-Cut Algorithm for Hierarchical Clustering}\label{alg:phiSparsestCut}
  \begin{algorithmic}[1]
    \State \textbf{Input}: An edge weighted graph $G = (V,E,w)$.
{   \State    $\{A,V \setminus A\} \gets$ cut with sparsity $\le \phi\cdot  \min\limits_{S \subset V}  w(S,V\setminus S)/(|S||V \setminus S|)$
    \State Recurse on  $G[A]$ and on $G[V \setminus A]$ to obtain trees $T_A$ and $T_{V\setminus A}$ \\
    \Return the tree whose root has two children, $T_A$ and $T_{V\setminus A}$.
    }
  \end{algorithmic}
\end{algorithm}

\begin{theorem}\footnote{For Dasgupta's function, this was already proved in \citet{CC16}  with a different constant. The present, independent proof, uses a different method.}
  \label{thm:SCapprox}
  For any graph $G=(V,E)$, and weight function $w : E \rightarrow \R_+ $,
  the $\phi$-sparsest-cut algorithm (Algorithm~\ref{alg:phiSparsestCut})
  outputs a solution of cost at 
  most $\frac{27}{4} \phi \opt$.  
\end{theorem}

\begin{proof}
Let $G = (V,E)$ be the input graph and $n$ denote the total number of vertices of $G$.
Let $T$ denote the tree output by the algorithm and $T^*$ be any arbitrary tree.
We will prove that $\cost(T) \le \frac{27}{4} \phi \cost(T^*)$. 
\footnote{The following paragraph bears similarities with the first part of 
the analysis of~\cite[Lemma 11]{Das:2016} but we obtain a more fine-grained
analysis by introducing a charging scheme.}

Recall that for an arbitrary tree $T_0$ and node $N$ of $T_0$, 
the vertices corresponding to the leaves of the subtree rooted
at $N$ is denoted by $V(N)$. 
Consider the node $N_0$ of $T^*$ that is the first node reached by the walk from the root that always goes to 
the child tree with the higher number of leaves, 
stopping when the subtree of $T^*$ rooted at $N_0$ contains fewer than $2n/3$ leaves.
The \emph{balanced cut} (BC) of $T^*$ 
is the cut $(V(N_0), V - V(N_0))$.
For a given node $N$ with children $N_1,N_2$, we say that the cut induced by $N$ is
the sum of the weights of the edges between that have one extremity in $V(N_1)$ 
and the other in $V(N_2)$.


Let $(A\cup C, B\cup D)$ 
be the cut induced by the root node $u$ of $T$, where $A,B,C,D$ are such that
$(A\cup B,C\cup D)$ 
is the {balanced cut}   of $T^*$. 
Since $(A\cup C,B\cup D)$  
is a $\phi$-approximate sparsest cut:
$$
\frac{w(A\cup C,B\cup D)}{|A\cup C|\cdot |B\cup D|} \le \phi 
\frac{w(A\cup B,C\cup D)}{|A\cup B|\cdot |C\cup D|}.
$$
By definition of $N_0$, $A\cup B$ and $C\cup D$ both have size in $[n/3,2n/3]$, 
so the product of their sizes is at least $(n/3)(2n/3)=2n^2/9$;  
developing $w(A\cup B,C\cup D)$ into four terms, we obtain
  \begin{eqnarray*}
  w(A\cup C,B\cup D)
&  \leq &\phi  \frac{9}{2n^2}   {|A\cup C| |B\cup D|} (w(A,C)+w(A,D)+w(B,C)+w(B,D) )\\
& \leq & \phi \frac{9}{2} [ \frac{|B\cup D|}{n}w(A,C)+w(A,D)+w(B,C) +  \frac{|A\cup C|}{n}w(B,D)  ],
\end{eqnarray*}
and so the cost induced by node $u$ of $T^*$ satisfies
  $$ 
  n \cdot w (A\cup C,B\cup D) \le \frac{9}{2} \phi |B\cup D| w(A,C) + \frac{9}{2} \phi |A\cup C| w (B,D) +\frac{9}{2} \phi n (w(A,D) + w(B,C)).
  $$
  To account for the cost induced by $u$, we thus assign a charge of $({9}/{2}) \phi |B\cup D| w(e)$ to each edge $e$ of $(A,C)$, 
  a charge of $({9}/{2}) \phi |A\cup C|w(e)$ to each edge $e$ of $(B,D)$, and a charge of $({9}/{2}) \phi n w(e)$ to each edge $e$ of $(A,D)$ or $(B,C)$.
  
  When we do this for every node $u$ of $T$, how much does each edge get charged? 
  \begin{lemma}\label{lemma:LCA}
  Let $G=(V,E)$ be a graph on $n$ nodes. We consider the above charging scheme for $T$ and $T^*$.
  Then, an edge $(v_1,v_2)\in E$ 
  gets charged at most $(9/2)\phi  \min ((3/2) |V(\lca_{T^*}(v_1,v_2))| ,n) w(e)$ 
  overall, where $\lca_{T^*}(v_1,v_2)$ denotes the lowest common ancestor of $v_1$ and $v_2$ in $T^*$. 
  \end{lemma}

  We temporarily defer the proof and first see how Lemma~\ref{lemma:LCA} implies the theorem.   
  Observe (as in~\cite{Das:2016}) 
  that   
  $\cost(T^*)=\sum_{\{u,v\}\in E} |V(\lca_{T^*}(u,v))| w(u,v)$.
  Thanks to Lemma~\ref{lemma:LCA}, when we sum charges assigned because 
  of every node $N$ of $T$, overall we obtain
  $$\cost(T)\leq  \frac{9}{2}\phi \sum_{\{v_1,v_2\}\in E} \frac{3}{2}|V(\lca_{T^*}(v_1,v_2))| 
  w(v_1,v_2)  =  \frac{27}{4}\phi  \cost(T^*).
  $$
\end{proof}

  \begin{proof}[Proof of Lemma~\ref{lemma:LCA}]

  The lemma is proved by induction on the number of nodes of the graph. 
  (The base case is obvious.) For the inductive step, consider the cut $(A\cup C,B\cup D)$ induced by the root node $u$ of $T$. 
  \begin{itemize}
  \item 
 Consider the edges that cross the cut. First, observe that edges of $(A,B)$ or of $(C,D)$ never get charged at all. Second, an edge $e=\{v_1,v_2\}$ of $(A,D)$ or of $(B,C)$ gets charged $({9}/{2}) \phi n w(e)$ when considering the cost induced by node $u$, and does not get charged when considering any other node of $T$. In $T^*$, edge $e$
 is separated by the cut $(A\cup B,C\cup D)$ induced by $N_0$, so the least common ancestor of $v_1$ and $v_2$ is the 
 parent node of $N_0$ (or above), and by definition of $N_0$ we have $|V(\lca_{T^*}(v_1,v_2))|\geq 2n/3$, hence the lemma holds for $e$.
    \item 
    An edge $e=\{ v_1,v_2\}$ of $G[A]\cup G[C]$ does not get charged when considering the cut induced by node $u$. Apply Lemma~\ref{lemma:LCA} to $G[A\cup C]$ for the tree $T^*_{A \cup C}$ defined as the subtree of $T^*$ induced by the vertices of $A \cup C$\footnote{note that $T^*_{A \cup C}$ is not necessarily the optimal tree for $G[A\cup C]$, which is why the lemma was stated in terms of every tree $T^*$, not just on the optimal tree.}. By induction, the overall charge to $e$ due to  the recursive calls for $G[A\cup C]$  is at most $(9/2)\phi \min((3/2) |V(\lca_{T^*_{A \cup C}}(v_1,v_2))| ,|A\cup C|) w(e).$ By definition of $T^*_{A \cup C}$, we have $|V(\lca_{T^*_{A \cup C}}(v_1,v_2))|\leq |V(\lca_{T^*}(v_1,v_2))|$, and $|A\cup C|\leq n$, so the lemma holds for $e$.
  \item 
  An edge $\{v_1,v_2\}$ of $(A,C)$ gets a charge of  $(9/2)\phi |B\cup D|  w(e)$ plus the total charge to $e$ coming from the recursive calls for $G[A\cup C]$ and the tree $T^*_{A \cup C}$. By induction  the latter is at most $$(9/2)\phi  \min((3/2) |V(\lca_{T^*_{A \cup C}}(v_1,v_2))| ,|A\cup C|) w(e)\leq (9/2)\phi  |A\cup C| w(e).$$ Overall the charge to $e$ is at most $(9/2)\phi n w(e)$. Since the cut induced by node $u_0$ of $T^*$ separates $v_1$ from $v_2$, we have $|V(\lca_{T^*}(v_1,v_2))|\geq 2n/3$, hence the lemma holds for $e$.
     For edges of $(B,D)$ or of $G[B]\cup G[D]$, a symmetrical argument applies.
    \end{itemize}
\end{proof}

\begin{remark}
  The recursive $\phi$-sparsest-cut algorithm
  achieves an $O(f_n \phi)$-approximation for any admissible
  cost function $f$, where 
  $f_n = \max_n f(n)/f(\lceil n/3 \rceil)$.
  Indeed,
  adapting the definition of the balanced cut as in~\cite{Das:2016}
  and rescaling the charge by a factor of $f_n$ 
  imply the result.
\end{remark}

 We complete our study of classical algorithms for hierarchical 
  clustering by showing that the standard agglomerative heuristics
  can perform poorly 
  (Theorems~\ref{T:sim:hard:sandc},~\ref{T:sim:hard:avg}).
  Thus, the sparsest-cut-based approach seems to be more reliable 
  in the worst-case.
  To understand better the success of the agglomerative
  heuristics, we restrict our attention
  to \structinputs (Section~\ref{S:perfectdata}), 
  and random graphs (Section~\ref{S:randominputs}),
  and show that in these contexts 
  these algorithms are efficient.

\section{Admissible Objective Functions and Algorithms for Random Inputs}
\label{S:randominputs}
\ifdraft
\newcommand{\hsbm}{\operatorname{HSBM}}

In this section, we initiate a \emph{beyond-worst-case} analysis of the
hierarchical clustering problem (see also Section~\ref{S:sub:robust}). We
study admissible objective functions in the context of random graphs that
have a natural hierarchical structure; for this purpose, we consider a
suitable generalization of the stochastic block model to hierarchical
clustering.

We show that, for admissible cost functions, an underlying ground-truth
cluster tree has optimal expected cost.  Additionally, for a subfamily of
admissible cost functions (called \emph{smooth}, see
Defn.~\ref{defn:smooth})
which includes the cost function introduced by
Dasgupta, we show the following: The cost of the ground-truth cluster 
tree is with high probability
sharply concentrated (up to a factor of $(1+o(1))$ around its
expectation), and so of cost at most $(1+o(1))\opt$. 
This is further evidence that optimising admissible cost
functions is an appropriate strategy for hierarchical clustering.

We also provide a simple algorithm based on the SVD based approach of
\cite{McS01} followed by a standard agglomerative heuristic  yields
a hierarchical clustering which is, up to a factor $(1+o(1))$, optimal
with respect to smooth admissible cost functions.

\subsection{A Random Graph Model For Hierarchical Clustering}

\label{sec:muble}

We describe the random graph model for hierarchical clustering, called
the hierarchical block model. This model has already been studied
earlier, e.g.,~\cite{LTAPP:2017}. However, prior work has mostly focused
on statistical hypothesis testing and exact recovery in some regimes. We
will focus on understanding the behaviour of admissible objective
functions and algorithms to output cluster trees that have almost optimal
cost in terms of the objective function.  

We assume that there are $k$ ``bottom''-level clusters that are then
arranged in a hierarchical fashion. In order to model this we will use a
similarity graph on $k$ nodes generated from an ultrametric (see
Sec.~\ref{SubSec:ultrametric}). 
There are $n_1, \ldots, n_k$ nodes in each of the $k$ clusters. Each edge
is present in the graph with a probability that is a function of the clusters
in which their endpoints lie and the underlying graph on $k$ nodes
generated from the ultrametric. The formal definition follows.

\begin{definition}[Hierarchical Stochastic Block Model (HSBM)] \label{defn:hsbm}
	A hierarchical stochastic block model with $k$ bottom-level clusters is
	defined as follows:
	\begin{itemize}
		\item Let $\tilde{G}_k = (\tilde{V}_k, \tilde{E}_k, w)$ be a graph
			generated from an ultrametric (see Sec.~\ref{SubSec:ultrametric}),
			where $|\tilde{V}_k| = k$ for each $e \in \tilde{E}_k$, $w(e)
			\in (0, 1)$.%
			\footnote{In addition to $\tilde{G}_k$ being generated from an
			ultrametric, we make the further assumption that the function $f :
			\reals_+ \rightarrow \reals_+$, that maps ultrametric distances to
			edge weights, has range $(0, 1)$, so that the weight of an edge can be
			interpreted as a probability of an edge being present. We rule out
			$w(e) = 0$ as in that case the graph is disconnected and each
			component can be treated separately.} 
			Let $\tilde{T}_k$ be a tree on $k$ leaves, let
			$\tilde{\calN}$ denote the internal nodes of $\tilde{T}$ and
			$\tilde{L}$ denote the leaves; let $\tilde\sigma : \tilde{L}
			\rightarrow [k]$ be a bijection. Let $\tilde{T}$ be generating
			for $\tilde{G}_k$ with weight function $\tilde{W} :
			\tilde{\calN} \rightarrow [0, 1)$ (see
			Defn.~\ref{defn:generating-tree}).  
		\item For each $i \in [k]$, let $p_i \in (0, 1]$ be such that $p_i >
			\tilde{W}(N)$, if $N$ denotes the parent of $\tilde\sigma^{-1}(i)$ in
			$\tilde{T}$.
		\item For each $i \in [k]$, there is a fixed constant $f_{i} \in
			(0, 1)$; furthermore $\sum_{i = 1}^k f_i = 1$.
	\end{itemize}
	Then a random graph $G = (V, E)$ on $n$ nodes with sparsity parameter
	$\alpha_n \in (0, 1]$ is defined as follows: $(n_1, \ldots, n_k)$ is drawn from
	the multinomial distribution with parameters $(n, (f_1, \ldots,
	f_k))$. Each vertex $i \in [n]$ is assigned a label $\psi(i) \in [k]$,
	so that exactly $n_j$ nodes are assigned the label $j$ for $j\in[k]$.
	 An edge $(i,
	j)$ is added to the graph with probability $\alpha_n p_{\psi(i)}$ if
	$\psi(i) = \psi(j)$ and with probability $\alpha_n \tilde{W}(N)$ if $\psi(i)
	\neq \psi(j)$ and $N$ is the least common ancestor of $\tilde\sigma^{-1}(i)$ and
	$\tilde\sigma^{-1}(j)$ in $\tilde{T}$.  The graph $G = (V, E)$ is returned without
	any labels.
\end{definition}

As the definition is rather long and technical, a few remarks are in
order. 
\begin{itemize}
	\item Rather than focusing on an arbitrary hierarchy on $n$ nodes, we assume
		that there are $k$ clusters (which exhibit no further hierarchy) and
		there is a hierarchy on these $k$ clusters. The model assumes that $k$ is
		fixed, but in future work, it may be interesting to study models where
		$k$ itself may be a (modestly growing) function of $n$. The condition
		$p_i > \tilde{W}(N)$ (where $N$ is the parent of $\tilde\sigma^{-1}(i)$ )
		ensures that nodes in cluster $i$ are strictly more likely to connect to
		each other than to node from any other cluster.
	\item The graphs generated can be of various sparsity, depending on
		the parameter $\alpha_n$. If $\alpha_n \in (0, 1)$ is a fixed
		constant, we will get dense graphs (with $\Omega(n^2)$ edges),
		however if $\alpha_n \rightarrow 0$ as $n \rightarrow \infty$,
		sparser graphs may be achieved. This is similar to the approach
		taken by~\citet{WO:2013} when considering random graph models
		generated according to graphons. 
\end{itemize}

We define the \emph{expected graph}, $\bar{G}$, which is a complete graph where
an edge $(i, j)$ has weight $p_{i, j}$ where $p_{i, j}$ is the probability with
which it appears in the random graph $G$. In order to avoid ambiguity, we
denote by $\Gamma(T; G)$ and $\Gamma(T; \bar{G})$ the costs of the cluster tree
$T$ for the unweighted (random) graph $G$ and weighted graph $\bar{G}$
respectively. Observe that due to linearity (see Eqns.~\eqref{eqn:costfunc}
and~\eqref{eqn:costfuncfull}), for any tree $T$ and any admissible cost
function, $\Gamma(T; \bar{G})  = \E{\Gamma(T; G)}$, where the expectation is
with respect to the random choices of edges in $G$ (in particular this holds
even when conditioning on $n_1, \ldots, n_k$).

Furthermore, note that $\bar{G}$ itself is generated from an ultrametric and
the generating trees for $\bar{G}$ are obtained as follows: Let $\tilde{T}_k$
be any generating tree for $\tilde{G}_k$, let $\hat{T}_1, \hat{T}_2, \ldots,
\hat{T}_k$ be any binary trees with $n_1, \ldots, n_k$ leaves respectively. Let
the weight of every internal node of $\hat{T}_i$ be $p_i$ and replace each leaf
$l$ in $\tilde{T}_k$ by $\hat{T}_{\tilde\sigma(l)}$. In particular, this last point
allows us to derive Proposition~\ref{thm:expectedmin}. We refer to any tree
that is generating for the expected graph $\bar{G}$ as a \emph{ground-truth
tree} for $G$. 

\begin{remark} \label{rem:ease}
	Although it is technically possible to have $n_i = 0$ for some $i$
	under the model, we will assume in the rest of the section that $n_i >
	0$ for each $i$. This avoids getting into the issue of degenerate
	ground-truth trees; those cases can be handled easily, but add no
	expository value.
\end{remark}


\subsection{Objective Functions and Ground-Truth Tree}

In this section, we assume that the graphs represent similarities. This
is clearly more natural in the case of unweighted graphs; however, all
our results hold in the dissimilarity setting and the proofs are
essentially identical.

%

\begin{proposition}\label{thm:expectedmin}
	Let $\Gamma$ be an admissible cost function.  Let $G$ be a graph generated
	according to an $\hsbm$ (See Defn.~\ref{defn:hsbm}). 
	 Let $\psi$ be the (hidden) function mapping the
	nodes of $G$ to $[k]$ (the bottom-level clusters).
	Let $T$ be a
	ground-truth tree for $G$ Then, 
	\[ \E{\Gamma({T}) ~|~ \psi} \le \min_{{T}'}{\E{\Gamma({T}') ~|~ \psi}}. \] 
	Moreover, for any tree ${T}'$,\ $\E{\Gamma({T}) ~|~ \psi} = \E{\Gamma({T}') ~|~ \psi}$ if and
	only if ${T}'$ is a ground-truth tree.
\end{proposition}
\begin{proof}
%
%
%

	As per Remark~\ref{rem:ease}, we'll assume that each $n_i > 0$ to
	avoid degenerate cases. Let $\bar{G}$ be a the expected graph, \ie
	$\bar{G}$ is complete and an edge $(i, j)$ has weight $p_{ij}$, the
	probability that the edge $(i, j)$ is present in the random graph $G$
	generated according to the hierarchical model. Thus, by definition of
	admissibility $\Gamma(T; \bar{G}) = \min_{T^\prime} \Gamma(T^\prime;
	\bar{G})$ if an only if $T$ is generating (see
	Defn.~\ref{defn:admissibleCF}). As ground-truth trees for $G$ are
	precisely the generating trees for $\bar{G}$; the result follows by
	observing that for any tree $T$ (not necessarily ground-truth)
	$\E{\Gamma(T; G)~|~\psi} = \Gamma(T; \bar{G})$, where the expectation is
	taken only over the random choice of the edges, by linearity of
	expectation and the definition of the cost function
	(Eqns.~\ref{eqn:costfunc} and~\ref{eqn:costfuncfull}).
\end{proof}

\newcommand{\lipsch}{smoothness\xspace}
\begin{defn} \label{defn:smooth}
	Let $\gamma$ be a cost function defined using the function $g(\cdot, \cdot)$ (see Defn.~\ref{defn:admissibleCF}). We say that the cost function $\Gamma$ (as defined in Eqn.~\ref{eqn:costfuncfull}) satisfies the \emph{\lipsch}
  property if 
  	\[g_{\max} := \max \{ g(n_1,n_2) ~|~n_1+n_2 =n\} = O\left(\frac{\kappa(n)}{n^2}\right),\] 
	where $\kappa(n)$ is the cost of a unit-weight clique of size $n$ under the
	cost function $\Gamma$.
\end{defn}
\begin{fact}
  The cost function introduced by~\cite{Das:2016} 
  satisfies the \lipsch property.
\end{fact}

\begin{theorem}\label{thm:optconcentrated}
	Let $\alpha_n =\omega( \sqrt{\log n/n})$. Let $\Gamma$ be an
	admissible cost function satisfying the \lipsch property
	(Defn.~\ref{defn:smooth}). Let $k$ be a fixed constant and $G$ be a
	graph generated from an $\hsbm$ (as per Defn.~\ref{defn:hsbm}) where
	the underlying graph $\tilde{G}_k$ has $k$ nodes and the sparsity
	factor is $\alpha_n$. Let $\psi$ be the (hidden) function mapping the
	nodes of $G$ to $[k]$ (the bottom-level clusters).  For any binary
	tree $T$ with $n$ leaves labelled by the vertices of $G$, the
	following holds with high probability:
	\[ \left|\Gamma(T) -\E{\Gamma(T)~|~\psi}\right| \leq
	o(\E{\Gamma(T) ~|~ \psi}). \] 
	The expectation is taken only over the random choice of edges.  In
	particular if $T^*$ is a ground-truth tree for $G$, then,
	with high probability, 
	\[ \Gamma(T^*) \le (1+o(1)) \min_{T'} \Gamma(T') = (1+o(1))\opt.\]
\end{theorem}
\begin{proof}
	Our goal is to show that for any fixed cluster tree $T'$ the cost is
	sharply concentrated around its expectation with an extremely high
	probability.  We then apply the union bound over all possible cluster
	trees and obtain that in particular the cost of $\OPT$ is sharply
	concentrated around its expectation. Note that there are at most $2^{c
	\cdot n \log n}$ possible cluster trees (including labellings of the
	leaves to vertices of $G$), where $c$ is a suitably large constant.
	Thus, it suffices to show that for any cluster tree $T'$ we have
	\begin{align*}
		\Pr{ \left|\Gamma(T') -\E{\Gamma(T')~\mid~\psi}\right| \geq
		o(\E{\Gamma(T')~\mid~\psi})    } \leq  \exp\left(-c^* n \log
		n\right), 
	\end{align*}
	where $c^* > c$.

	Recall that for a given node $N$ of $T'$ with children $N_1,N_2$, we
	have $\gamma(N) = w(V(N_1),V(N_2)) \cdot g(|V(N_1)|,|V(N_2)|)  $ and
	$\Gamma(T') = \sum_{N \in T'} \gamma(N)$ (see
	Eqns.~\eqref{eqn:costfunc} and~\eqref{eqn:costfuncfull}). Let
	$Y_{i,j}=\mathbf{1}_{(i,j)\in E}$ for all $1\leq i,j\leq n$ and
	observe that $\{ Y_{i,j} | i < j\}$ are independent and
	$Y_{i,j}=Y_{j,i}$.  Furthermore, let
	$Z_{i,j}=g(|V(\choc{N^{i,j}})|,|V(\chtc{N^{i,j}})|) \cdot Y_{i,j}$,
	where $N^{i,j}$ is the node in $T'$ separating nodes $i$ and $j$ and
	$\choc{N^{i, j}}$ and $\chtc{N^{i, j}}$ are the two children of $N^{i,
	j}$. We can thus write
	\begin{align}\label{eq:runoutofnames}
		\Gamma(T') &= \sum_{N \in T'} g(|V(\choc{N})|,|V(\chtc{N})|)
		\sum_{\substack{i\in V(\choc{N}) \\ j\in V(\chtc{N})}} Y_{i,j} \\
		&= \sum_{N \in T'}\sum_{\substack{i\in V(\choc{N})\\ j\in
		V(\chtc{N})}} Z_{i,j} \\
		&= \sum_{i < j} Z_{i,j}, 
	\end{align} 
	where we used that every potential edge ${i,j}$, $i\neq j$ appears in
	exactly one cut and that $Z_{i,j}=Z_{j,i}$.  Observe that $\sum_{i <
	j} Z_{i,j}$ is a sum of independent random variables. Assume that the
	following claim holds.
	\begin{claim} 
          \label{cl:onstuff}
		Let $w_{\min}=\Omega(1)$ be the minimum weight in $\tilde{T}_k$,
		the tree generating tree for $\tilde{G}_k$ (see
		Defn.~\ref{defn:hsbm}), \ie $w_{\min} = \min_{N \in \tilde{T}_k}
		\tilde{W}(N)$) and recall that $g_{\max}= \max \{ g(n_1,n_2)
		~|~n_1+n_2 =n\}$. We have
		\begin{enumerate}
			\item	$\E{\Gamma(T')~|~\psi} \geq \kappa(n) \cdot \alpha_n\cdot
				w_{\min}$
			\item $\sum_{i<j} g(|V(\choc{N^{i,j}})|,|V(\chtc{N^{i,j}})|)^2 \leq
				g_{\max} \cdot \kappa(n)	$
		\end{enumerate}
	\end{claim}

	We defer the proof to later and first finish the proof of
	Theorem~\ref{thm:optconcentrated}.  We will make use of the slightly
	generalized version of Hoeffding bounds (see~\cite{H63}).  For $X_1,
	X_2, \dots, X_m$ independent random variables satisfying $a_i \leq X_i
	\leq b_i$ for $i\in [n]$.  Let $X=\sum_{i=1}^m X_i$, then for any $t
	>0$
	\begin{align}\label{eq:hoeff}
		\Pr{|X - \E{X}| \geq t } \leq \exp\left(-\frac{2t^2}{\sum_{i=1}^m
		(b_i - a_i)^2}\right).
	\end{align}
	By assumption, there exists a function $y_n:\mathbb{N}\rightarrow
	\mathbb{R}_+$ such that $ \alpha_n =\omega\left( y_n \cdot
	\sqrt{\frac{\log n}{  n}}\right)$ with $y_n=\omega(1)$  for all $n$.
	We  apply \eqref{eq:hoeff} with $t= \E{\Gamma(T')~\mid~\psi}\cdot
	\frac{y_n\cdot \sqrt{\frac{\log n}{
		n}}}{\alpha_n}=o(\E{\Gamma(T')~\mid~\psi})$ and derive 
	\begin{align*}
		&\Pr{ \left|\Gamma(T') -\E{\Gamma(T')~\mid~\psi} \right| \leq
		\E{\Gamma(T')~\mid~\psi} \cdot \frac{y_n\cdot \sqrt{\frac{\log n}{
			n}}}{\alpha_n}   } \geq \\
		&\phantom{0000}\geq 1-
		\exp\left(-\frac{2\left(\E{\Gamma(T')~\mid~\psi}  \cdot
		\frac{y_n\cdot \sqrt{\frac{\log n}{  n}}}{\alpha_n}  \right)^2}{
			\sum_{i<j} g(|V(N^{i,j}_1)|,|V(N^{i,j}_2)|)^2 }\right)\\
		&\phantom{0000}\geq 1- \exp\left(-\frac{2 \cdot\kappa(n)\cdot
		w_{\min}^2 \cdot y_n^2\cdot  \log n }{g_{\max}\cdot n }\right)\\
		&\phantom{0000}\geq 1- \exp\left(-c^* \cdot n \log n\right),
	\end{align*}
	where the last inequality follows by assumption of the lemma and since
	$y_n =\omega(1)$. 
        \end{proof}
        We now turn to the proof of Claim~\ref{cl:onstuff}.
        \begin{proof}[Proof of Claim~\ref{cl:onstuff}]
          Note that for any two vertices
	$i, j$ of $G$, the edge $(i, j)$ exists in $G$ with probability at
	least $\alpha_n \cdot w_{\min}$. Thus, we have
	\begin{align*}
		\E{\Gamma(T') ~|~ \psi} &=\sum_{N \in T'} g(|V(\choc{N})|,|V(\chtc{N})|)
		\sum_{\substack{i\in V(\choc{N})\\ j\in V(\chtc{N})}} \alpha_n
		\cdot w(i,j) \\
		&\geq w_{\min} \cdot \alpha_n \cdot \sum_{N \in T'}
		g(|V(\choc{N})|,|V(\chtc{N})|) |V(\choc{N})| \cdot |V(\chtc{N})| \\
		&=  w_{\min} \cdot \alpha_n \cdot \kappa(n),
	\end{align*}
	where we made use of Eqn.~\eqref{eq:honig}.  Furthermore, we have
	\begin{align*}
		\sum_{i<j} g(|V(\choc{N^{i,j}})|,&|V(\chtc{N^{i,j}})|)^2=\\ 
		&=\sum_{N \in T'} g(|V(\choc{N})|,|V(\chtc{N})|)^2 \cdot |V(\choc{N})|\cdot |V(\chtc{N}))|
\\&\leq g_{\max} \sum_{N \in T'}  g(|V(\choc{N})|,|V(\choc{N})|) \cdot |V(\choc{N})|\cdot |V(\chtc{N}))| \\
&=g_{\max} \cdot \kappa(n),	
\end{align*}
where we made use of Eqn.~\eqref{eq:honig}.
%
%
%
\end{proof}

\subsection{Algorithm for Clustering in the $\hsbm$}

In this section, we provide an algorithm for obtaining a hierarchical
clustering of a graph generated from an $\hsbm$. The algorithm is quite
simple and combines approaches that are used in practice for 
hierarchical clustering: SVD projections and agglomerative heuristics.
See Algorithm~\ref{alg:random} for a complete description.

\begin{theorem} \label{thm:random:recovery}
	Let $\alpha_n =\omega( \sqrt{\log n/n})$.  Let $\Gamma$ be an
	admissible cost function (Defn.~\ref{defn:admissibleCF}) satisfying
	the \lipsch property (Defn~\ref{defn:smooth}). Let $k$ be a fixed constant and $G$ be a
	graph generated from an $\hsbm$ (as per Defn.~\ref{defn:hsbm}) where
	the underlying graph $\tilde{G}_k$ has $k$ nodes and the sparsity
	factor is $\alpha_n$. 
        Let $T$ be a
	ground-truth tree for $G$. With high probability,
	Algorithm~\ref{alg:random} with parameter $k$ on graph $G$ outputs a
	tree $T'$ that satisfies $\Gamma(T) \leq (1+o(1)) \opt$.
\end{theorem}


\begin{algorithm}[H]
  \caption{Agglomerative Algorithm for Recovering
    Ground-Truth Tree of an HSBM Graph}
  \label{alg:random}
  \begin{algorithmic}[1]
    \State \textbf{Input:} Graph $G=(V,E)$ generated from
    an $\hsbm$.
    \State \textbf{Parameter:} A constant $k$.
	  \State Apply (SVD) projection algorithm of~\cite[Thm. 12]{McS01}
	  with parameters $G$, $k$, $\delta=|V|^{-2}$, to get $\zeta(1),
	  \ldots, \zeta(|V|) \in \reals^{|V|}$ for vertices in $V$, where
	  $\mathrm{dim}(\mathrm{span}(\zeta(1),\ldots, \zeta(|V|))) = k$. \label{random:step:SVD} 
	 \State Run the single-linkage algorithm on the points $\{\zeta(1),
	  \ldots, \zeta(|V|) \}$ until there are exactly $k$ clusters. Let
	  $\calC = \{C^\zeta_1,\ldots,C^\zeta_{k}\}$ be the clusters (of
	  points $\zeta(i)$) obtained. Let $C_i \subseteq V$ denote the set of
	  vertices corresponding to the cluster $C^\zeta_i$.
	  \label{random:step:link}
    \While{there are at least two clusters in $\calC$}\label{random:step:while}
    \State\label{random:step:sparse}
    Take the pair of clusters $C_i,C_j$ of $\calC$ that maximizes
	  $\frac{\mathrm{cut}(C_i,C_j)}{|C_i|\cdot|C_j|}$
	  \State $\calC \gets \calC ~\setminus ~\{C_i\} ~\setminus~ \{C_j\} ~\cup ~\{C_i \cup C_j\}$
    \EndWhile\label{random:step:endwhile}
	 \State The sequence of merges in the while-loop
	  (Steps~\ref{random:step:while} to \ref{random:step:endwhile})
	  induces a hierarchical clustering tree on $\{C_1, \ldots, C_k\}$,
	  say $T^\prime_k$ with $k$ leaves (represented by $C_1, \ldots,
	  C_k$). Replace each leaf of $T^\prime_k$ by an arbitrary binary tree
	  on $|C_k|$ leaves labelled according to the vertices $C_k$
	  to obtain $T$.  \label{random:step:thatonetree} 
    \State Repeat the algorithm $k'=2k\log n$ times. Let $T^1,\dots T^{k'}$ be the corresponding hierarchical clustering trees.
	  \State \textbf{Output:} Tree $T^i$ (out of the $k^\prime$ candidates) that minimises $\Gamma(T_i)$.
  \end{algorithmic}
\end{algorithm}

\begin{remark}
	In an $\hsbm$, $k$ is a fixed constant.  Thus, even if $k$ is not
	known in advance, one can simply run the Algorithm~\ref{alg:random}
	with all possible different values (constantly many) and return the
	solution with the minimal cost $\Gamma(T)$.
\end{remark}



Let $G=(V,E)$ be the input graph generated according to an $\hsbm$. Let
$T$ be the tree output by Algorithm~\ref{alg:random}.  We divide the
proof into two claims that correspond to the outcome of
Step~\ref{random:step:SVD} and the while-loop
(Steps~\ref{random:step:while} to \ref{random:step:endwhile}) of
Algorithm~\ref{alg:random}.

%
%
%

We use a result of \cite{McS01} who considers a random graph model with
$k$ clusters that is (slightly) more general than the $\hsbm$ considered
here. The difference is that there is no hierarchical structure on top of
the $k$ clusters in his setting; on the other hand, his goal is also
simply to identify the $k$ clusters and not any hierarchy upon them. The
following theorem is derived from~\cite{McS01} (Observation 11 and a simplification of
Theorem 12).

\begin{theorem}[{\cite{McS01}}] \label{thm:sherry}
	Let $s$ be the size of the smallest cluster (of the $k$ clusters) and $\delta$ be the confidence parameter.
	Assume that for all $u,v$ belonging to different clusters with with
	adjacency vectors $\mathbf{u},\mathbf{v}$ (\ie $u_i$ is $1$ if
	the edge $(u, i)$ exists in $G$ and $0$ otherwise) satisfy
	\[ \twonorm{\E{\mathbf{u}} -\E{\mathbf{v}}}^2 \geq c \cdot {k}
	\cdot \left({n/s}+{\log (n/\delta)}\right)\]
	for a large enough constant $c$, where $\E{\mathbf{u}}$ is the
	entry-wise expectation.  Then, the algorithm of \citet[Thm. 12]{McS01}
	with parameters $G,k, \delta$ projects the columns of the adjacency matrix of
	$G$
	to points $\{\zeta(1), \ldots, \zeta(|V|) \}$ in a
	$k$-dimensional subspace of $\mathbb{R}^{|V|}$ such that the following
	holds w.p. at least $1-\delta$ over the random graph $G$ and with
	probability $1/k$ over the random bits of the algorithm.  There exists
	$\eta > 0$ such that for any $u$ in the $i$th  cluster and $v$ in
	the $j$th cluster:
   \begin{enumerate}
		\item if $i=j$ then $\twonorm{\zeta(u) -\zeta(v)}^2 \le \eta$;
    	\item if $i \neq j$ then $\twonorm{\zeta(u)-\zeta(v)}^2 > 2\eta$, 
    \end{enumerate}
\end{theorem}
Recall that $\psi : V \rightarrow [k]$ is the (hidden) labelling assigning each
vertex of $G$ to one of the $k$ bottom-level clusters. Let $C^*_i = \{ v
\in V ~|~ \psi(v) = i \}$. Recall that $n_i=|V(C^*_i)|$.  Note that the
algorithm of \cite[Thm. 12]{McS01} might fail for two reasons. The first
reason is that the random choices by the algorithm yield an incorrect
clustering. This happens w.p. at most $1/k$ and we can simply repeat the
algorithm sufficiently many times to be sure that at least once we get
the desired result, \ie the projections satisfy the conclusion of
Thm.~\ref{thm:sherry}. Claims~\ref{cl:proj} and~\ref{cl:link} show that
in this case, Steps \ref{random:step:while} to \ref{random:step:endwhile}
of Alg.~\ref{alg:random} produce a tree that has cost close to optimal.
Ultimately, the algorithm simply outputs a tree that has the least cost
among all the ones produced (and one of them is guaranteed to have cost
$(1 + o(1))\opt$) with high probability.

The second reason why the McSherry's algorithm may fail is that the
generated random graph $G$ might ``deviate'' too much from its
expectation. This is controlled by the parameter $\delta$ (which we set
to $1/|V|^2$). Deviations from expected behaviour will cause our
algorithm to fail as well. We bound this failure probability in terms of
two events. Let $\bar{\mathcal{E}}_1$ be the event that there exists $i$,
such that $n_i < f_i n/2$, \ie at least one of the bottom-level
clusters has size that is not representative. This event occurs with a
very low probability which is seen by a simple application of the
Chernoff-Hoeffding bound, as $\E{n_i} = f_i n$. Note that
$\bar{\mathcal{E}}_1$ depends only on the random choices that assign
labels to the vertices according to $\psi$ (and not on random choice of
the edges). Let $\mathcal{E}_1$ be the complement of
$\bar{\mathcal{E}}_1$. If $\mathcal{E}_1$ holds the term $n/s$ that
appears in Thm.~\ref{thm:sherry} is a constant. The second bad event is
that McSherry's algorithm fails due to the random choice of edges. This
happens with probability at most $\delta$ which we set at $\delta =
\frac{1}{|V|^2}$. We denote the complement of this event $\mathcal{E}_2$.
Thus, from now on we assume that both ``good'' events $\mathcal{E}_1$ and $\mathcal{E}_2$
occur, allowing Alg.~\ref{alg:random} to fail if either of them don't
occur. 

%
  %
  In order to prove Theorem~\ref{thm:random:recovery} we establish the
  following claims.  
  
\begin{claim}\label{cl:proj}
	Let $\alpha_n=\omega( \sqrt{\log n/n})$.  Let $G$ be generated by an
	$\hsbm$. Let $C_1^*, \ldots, C_k^*$ be the hidden bottom-level
	clusters, \ie $C_i^* = \{v ~|~ \psi(v) = i \}$.  Assume that events
	$\mathcal{E}_1$ and $\mathcal{E}_2$ occur. With probability at least
	$1/k$, the clusters obtained after Step~\ref{random:step:link}
	correspond to the assignment $\psi$, \ie there exists a permutation
	$\pi : [k] \rightarrow [k]$, such that $C_j = C^*_{\pi(j)}$.
\end{claim}
\begin{proof}
	The proof relies on Theorem~\ref{thm:sherry}. As we
	know that the event $\mathcal{E}_1$ occurs, we may conclude that $ s =
	\min_{i} n_i \geq \frac{n}{2} \min_{i} f_i$. Thus, $n/s \leq
	\frac{2}{f_{\min}}$, where $f_{\min} = \min_i f_i$ and hence $n/s$ is
	bounded by some fixed constant. 

	Let $u,v$ be two nodes such that $i = \psi(u) \neq \psi(v) = j$. Let
	$\mathbf{u}$ and $\mathbf{v}$ denote the random variables
	corresponding to the columns of $u$ and $v$ in the adjacency matrix of
	$G$. Let $q = \tilde{W}(N)$ where $N$ is the
	$\lca_{\tilde{T}_k}(\tilde\sigma^{-1}(i), \tilde\sigma^{-1}(j))$ in $\tilde{T}_k$,
	the generating tree for $\tilde{G}_k$ used in defining the $\hsbm$.
	Assuming $\mathcal{E}_1$ and taking expectations only with respect to
	the random choice of edges, we have:
 	\begin{align*}
		\twonorm{\E{\mathbf{u}~|~\psi,\mathcal{E}_1} - \E{ \mathbf{v}~|~\psi,\mathcal{E}_1 }}^2 &\geq n_i  \alpha_n^2
		(p_i - q)^2 + n_j \alpha_n^2 (p_j - q)^2 = \Omega(\alpha_n^2 n) =
		\omega( \log n)
 	\end{align*}
	Above we used that $p_i - q > 0$ and $n_i = \Omega(n)$ for each $i$.

	Note that for $\delta = \frac{1}{n^2}$, this satisfies the condition
	of Theorem~\ref{thm:sherry}. Since, we are already assuming that
	$\mathcal{E}_2$ holds, the only failure arises from the random coins
	of the algorithm. Thus, with probability at least $1/k$ the
	conclusions of Theorem~\ref{thm:sherry} hold. In the rest of the proof
	we assume that the following holds: There exists $\eta > 0$ such that for any pair of nodes $u,v$ we have
   \begin{enumerate}
		\item if $\psi(u) = \psi(v)$ then $\twonorm{\zeta(u) -\zeta(v)}^2
			\le \eta$;
		\item if $\psi(u) \neq \psi(v)$ then $\twonorm{\zeta(u)-\zeta(v)}^2
			> 2\eta$.
    \end{enumerate}

	Therefore, any linkage algorithm, \eg single linkage (See
	Alg.~\ref{alg:linkage4sim}), performing merges starting from the set
	$\{\zeta(1), \ldots, \zeta(n) \}$ until there are $k$ clusters will
	merge clusters at a distance of at most $\eta$ and hence, the clusters
	obtained after Step~\ref{random:step:link} correspond to the
	assignment $\psi$. This yields the claim.  
\end{proof}

\begin{claim} \label{cl:link}
	Let $\alpha_n=\omega( \sqrt{\log n/n})$. Let $G$ be generated
	according to an $\hsbm$ and let $T^*$ be a ground-truth tree for $G$.
	Assume that events $\mathcal{E}_1$ and $\mathcal{E}_2$ occur, and that
	furthermore, the clusters obtained after Step~\ref{random:step:link}
	correspond to the assignment $\psi$, \ie there exists a permutation
	$\pi : [k] \rightarrow [k]$ such that for each $v \in C_i$, $\psi(v) =
	\pi(i)$.  Then, the  sequence of merges in the while-loop
	(Steps~\ref{random:step:while} to \ref{random:step:endwhile}) followed
	by Step~\ref{random:step:thatonetree} produces w.h.p. a tree $T$ such
	that $\Gamma(T) \leq (1+o(1))OPT$.
\end{claim}
\begin{proof}
	For simplicity of notation, we will assume that $\pi$ is the identity
	permutation, \ie the algorithm has not only identified the true
	clusters correctly but also guessed the correct label. This only makes
	the notation less messy, though the proof is essentially unchanged.

	Let $N$ be some internal node of any  generating tree $\tilde{T}_k$ 
        and let $S_1 = V(\choc{N})$
	and $S_2 = V(\chtc{N})$. Note that both $S_1$ and $S_2$ are a disjoint union
	of some of the clusters $\{C_1, \ldots, C_k \}$. Then notice that for any $u
	\in S_1$ and $v \in S_2$, the probability that the edge $(u, v)$ exists in
	$G$ is $\alpha_n \tilde{W}(N)$. Thus, conditioned on $\mathcal{E}_1$ and $\psi$, 
        we have
	\begin{align*}
		\E{ \frac{\mathrm{cut}(S_i, S_j)}{|S_i| |S_j|}} = \alpha_n \tilde{W}(N)
	\end{align*}
	For now, let us assume that Alg.~\ref{alg:random} makes merges in
	Steps~\ref{random:step:while} to~\ref{random:step:endwhile} based on
	the true expectations instead of the empirical estimates. Then
	essentially, the algorithm is performing any linkage algorithm 
        (\ie average, single, or complete-linkage) on a
	ground-truth input and hence is guaranteed to recover the generating
	tree (see Theorem.~\ref{T:linkageAlgs:Perfectdata}). 

 	To complete the proof, we will show the following: For any partition $C_1,
 	\ldots, C_k$ of $V$ satisfying $\min_i |C_i| \geq \frac{n}{2} \min_{i}f_i$,
 	and for any $S_1, S_2, S_1^\prime, S_2^\prime$, where $S_1$ and $S_2$ (and
 	$S_1^\prime, S_2^\prime$) are disjoint and are both unions of some cluster
 	$\{C_1, \ldots, C_k \}$. and $i^\prime \neq j^\prime$, with probability at
 	least $1 - 1/n^3$, the following holds:
 	\begin{align}
		\left| \E{\frac{\mathrm{cut}(S_1, S_2)}{|S_1| \cdot |S_2|}} -  \frac{\cut(S_1,
		S_2)}{|S_1| \cdot |S_2|}\right| \leq \frac{1}{\log n}~\label{eqn:cutconc}
 	\end{align}
 	Note that as the probability of any edge existing is $\Omega(\alpha_n)$, it
 	must be the case that $\E{\frac{\cut(S_1, S_2)}{|S_1||S_2|}} =
 	\Omega(\alpha_n)$. Furthermore, since $|S_1|, |S_2| = \Omega(n)$ by the
 	Chernoff-Hoeffding bound for a single pair $(S_1, S_2)$ the above holds with
 	probability $\exp( - c n^2 \alpha_n / \log n)$. Thus, even after taking a
 	union bound over $O(k^n)$ possible partitions and $2^{O(k)}$ possible pairs
 	of sets $(S_1, S_2)$ that can be derived from said partition,
 	Eqn.~\ref{eqn:cutconc} holds with high probability.

	To complete the proof, we will show the following: We will assume that the
	algorithm of McSherry has identified the correct clusters at the bottom
	level, \ie $\mathcal{E}_2$ holds, and that Eqn.~\ref{eqn:cutconc} holds. 

	We restrict our attention to sets $S_1, S_2$ that are of the form that $S_1
	= V(\choc{N})$ and $S_2 = V(\chtc{N})$ for some internal node $N$ of some
	generating tree $\tilde{T}_k$ of the graph $\tilde{G}_k$ used to generate
	$G$. Then for pairs $(S_1, S_2)$ and $(S^\prime_1, S^\prime_2)$ both of this
	form, the following holds: there exists $\frac{3}{\log n} > \eta > 0$ such
	that 
	\begin{enumerate}
		\item If $\E{\frac{\cut(S_1, S_2)}{|S_1| \cdot |S_2|}} =
			\E{\frac{\cut(S^\prime_1, S^\prime_2)}{|S^\prime_1| \cdot |S^\prime_2|}}$, then $\left|
			\frac{\cut(S_1, S_2)}{|S_1| \cdot |S_2|} - \frac{\cut(S^\prime_1,
			S^\prime_2)}{|S^\prime_1| \cdot |S^\prime_2|}\right| \leq \eta$
		\item If $\E{\frac{\cut(S_1, S_2)}{|S_1| \cdot |S_2|}} \neq
			\E{\frac{\cut(S^\prime_1, S^\prime_2)}{|S^\prime_1| \cdot |S^\prime_2|}}$, then $\left|
			\frac{\cut(S_1, S_2)}{|S_1| \cdot |S_2|} - \frac{\cut(S^\prime_1,
			S^\prime_2)}{|S^\prime_1| \cdot |S^\prime_2|}\right| > 2 \eta$
	\end{enumerate}
	Note that the above conditions are enough to ensure that the algorithm
	performs the same steps as with perfect inputs, up to an arbitrary
	choice of tie-breaking rule. Since Theorem~\ref{T:linkageAlgs:Perfectdata}
        is true no matter the tie breaking rule chosen, the proof follows since the two above
	conditions hold with probability at least $\exp( - c n^2 \alpha_n / \log n)$. 
\end{proof}

We are ready to prove Theorem~\ref{thm:random:recovery}.
\begin{proof}[Proof of Theorem~\ref{thm:random:recovery}]
Conditioning on $\mathcal{E}_1$ and $\mathcal{E}_2$ which occur w.h.p.
we get from
 Claims~\ref{cl:proj} and \ref{cl:link} that w.p. at least $1/k$ the tree $T_i$ obtain in step \ref{random:step:thatonetree}
fulfills $\Gamma(T_i) \leq (1+o(1))OPT$.
It is possible to boost this probability by running Algorithm~\ref{alg:random} multiple times.
Running it $\Omega(k\log n)$ times 
and taking the tree with the smallest $\Gamma(T_i)$ yields the result.
\end{proof}
  
\begin{remark}
  It is worth mentioning that one of the trees $T_i$ computed by the algorithm is w.h.p. the ground-truth tree $T^*$.
  If one desires to recover that tree, then this is possible by verifying for each candidate tree
  with minimal $T_i$ whether is is indeed generating.
\end{remark}

\fi

\section{Dissimilarity-Based Inputs: Approximation Algorithms}
\label{S:dissim:approx}
In this section, we consider general dissimilarity inputs and
admissible objective for these inputs. 
For ease of exposition, we focus on an particular admissible 
objective function for dissimilarity inputs. 
Find $T$ maximizing the 
value function corresponding to Dasgupta's cost function of 
Section~\ref{S:sim:worstcase}:
$\val(T) = \sum_{N \in T} \val(N)$ where for each node $N$ of $T$ with children $N_1,N_2$, 
$\val(N) = w(V(N_1),V(N_2)) \cdot V(N)$.
This
optimization problem is NP-Hard \cite{Das:2016},
hence we focus on approximation 
algorithms.

We show (Theorem~\ref{T:avgapprox}) that average-linkage 
achieves a $2$ approximation for the problem.
We then
introduce a simple algorithm
 based on \emph{locally-densest cuts} and show (Theorem~\ref{T:dis:betterapprox}) that it 
achieves a $3/2+\eps$ approximation for the problem.

We remark that our proofs show that for any admissible objective function,
those algorithms have approximation guarantees, but 
the approximation guarantee depends 
on the objective function. 

We start with the following elementary upper bound on $\opt$.
\begin{fact}
  \label{F:optcost}
  For any graph $G=(V,E)$, and weight function $w : E \rightarrow \R_+ $,
  we have $\opt \le n \cdot \sum_{e \in E} w(e)$.
\end{fact}

\subsection{Average-Linkage}
We show that average-linkage is a 2-approximation
 in the dissimilarity setting.
\begin{theorem}
  \label{T:avgapprox}
  For any graph $G=(V,E)$, and weight function $w : E \rightarrow \R_+ $,
  the average-linkage algorithm (Algorithm~\ref{alg:avglinkage}) 
  outputs a solution of value at 
  least $n \sum_{e \in E} w(e)/2 \ge \opt/2.$
\end{theorem}
\begin{algorithm}
  \caption{Average-Linkage Algorithm for Hierarchical Clustering (dissimilarity setting)}\label{alg:avglinkage}
  \begin{algorithmic}[1]
    \State \textbf{Input:} Graph $G=(V,E)$ with edge weights 
    $w: E \mapsto \R_+$
    \State Create $n$ singleton trees. 
    \While{there are at least two trees  
    }
    \State Take trees roots $N_1$ and $N_2$ minimizing $\sum\limits_{x\in V(N_1),y\in V(N_2)} w(x,y)/(|V(N_1)||V(N_2)|)$
    \label{s:avglinkage:merge}
    \State Create a new tree with root $N$ and children $N_1$ and $N_2$ 
    \EndWhile\\
    \Return the resulting binary tree T
  \end{algorithmic}
\end{algorithm}

When two trees  are chosen at Step~\ref{s:avglinkage:merge}
of Algorithm~\ref{alg:avglinkage}, we say that they are 
\emph{merged}. We say that all the trees considered at the beginning of an iteration of the while loop
are the trees that are \emph{candidate for the merge} or simply
the \emph{candidate trees}.

We first show the following lemma and then prove the theorem.
\begin{lemma}
  \label{L:structavg}
  Let $T$ be the output tree and $A,B$ be the children of the root.
  We have,
  $$ \frac{\cut(V(A),V(B))}{|V(A)| \cdot |V(B)|} \ge 
  \frac{w(V(A))}{|V(A)|\cdot(|V(A)|-1)} +
  \frac{w(V(B))}{|V(B)|\cdot(|V(B)|-1)}.$$
\end{lemma}
\begin{proof}
  Let $a = |V(A)|(|V(A)|-1)/2$ and $b = |V(B)|(|V(B)|-1)/2$.
  For any node $N_0$ of $T$, let $\cho(N_0)$ and $\cht(N_0)$ 
  be the two children of $N_0$.
  We first consider the subtree $T_A$ of $T$ rooted at $A$. 
  We have
  $$\begin{cases}
 w(V(A)) &=  \sum_{A_0 \in T_A} \cut(V(\text{child}_1(A_0)),V(\text{child}_2(A_0))),\\
 a           &=  \sum_{A_0 \in T_A} |V(\text{child}_1(A_0))|\cdot 
  |V(\text{child}_2(A_0))|.
  \end{cases}
  $$
By an averaging argument, 
  there exists $A'\in T_A$ with 
  children $A_1,A_2$ such that 
  \begin{equation}
    \label{eq:avginA}
      \frac{\cut(V(A_1),V(A_2))}{|V(A_1)|\cdot |V(A_2)|} 
      \ge \frac{w(V(A))}{a}.
  \end{equation}
  We now consider the iteration of the while loop
  at which the algorithm merged 
  the trees $A_1$ and $A_2$.
  Let $A_1,A_2,\ldots,A_k$ and $B_1,B_2,\ldots,B_\ell$ be the trees
  that were candidate for the merge at that iteration, and such that
  $V(A_i) \cap V(B) = \emptyset$ and $V(B_i) \cap V(A) = \emptyset$.
  Observe that the leaves sets of those trees form a partition of 
  the sets $V(A)$ and $V(B)$, so
    $$\begin{cases}
\cut(A,B)&= \sum_{i,j} \cut(V(A_i),V(B_j)),\\
  |V(A)| \cdot |V(B)|&= \sum_{i,j} |V(A_i)|\cdot |V(B_j)|.
  \end{cases}
  $$
By an averaging argument again, 
  there exists $A_i,B_j$ such that 
  \begin{equation}\label{eq:avginAB}
    \frac{\cut(V(A_i),V(B_j))}{|V(A_i)| \cdot |V(B_j)|} \le 
    \frac{\cut(V(A),V(B))}{|V(A)|\cdot |V(B)|}.
  \end{equation}
  Now, since the algorithm merged $A_1,A_2$ rather than $A_i,B_j$,
  by combining Eq.~\ref{eq:avginA} and~\ref{eq:avginAB}, we have 
  $$\frac{w(V(A))}{a} \le 
  \frac{\cut(V(A_1),V(A_2))}{|V(A_1)|\cdot |V(A_2)|} 
  \le \frac{\cut(V(A_i),V(B_j))}{|V(A_i)| \cdot |V(B_j)|}\le 
  \frac{\cut(V(A),V(B))}{|V(A)| \cdot |V(B)|}.$$
  Applying the same reasoning to $B$ and taking the sum 
  yields the lemma.
\end{proof}

\begin{proof}[Proof of Theorem~\ref{T:avgapprox}]
  We proceed by induction on the number of the nodes $n$ in the graph.
  Let $A,B$ be the children of the root of the output tree $T$.
  By induction,
  \begin{align}\label{eq:twoapprox}
  \val(T)\geq (|V(A)|+|V(B)|) \cdot \cut(V(A),V(B)) + \frac{|V(A)|}{2} w(V(A)) + 
  \frac{|V(B)|}{2} w(V(B)).
  \end{align}
  Lemma~\ref{L:structavg} implies  
  $(|V(A)|+|V(B)|) \cdot \cut(V(A),V(B)) \ge |V(B)| w(V(A)) + |V(A)| 
  w(V(B))$. Dividing both sides by 2 and plugging it into \eqref{eq:twoapprox} yields
  $$\val(T)\geq \frac{|V(A)|+|V(B)|}{2} \cut(V(A),V(B)) + \frac{|V(A)| + 
    |V(B)|}{2} (w(V(A)) + w(V(B))).$$
  Observing that $n = |V(A)| + |V(B)|$ and combining
  $\sum_{e \in E} w(e) = \cut(V(A),V(B)) + w(V(A)) + w(V(B))$ with
  Fact~\ref{F:optcost} completes the proof.
\end{proof}

\subsection{A Simple and Better
  Approximation Algorithm for 
  Worst-Case Inputs} 
In this section, we introduce a very simple algorithm (Algorithm~\ref{alg:greedycut}) that achieves 
a better approximation guarantee.
The algorithm follows a divisive approach by recursively computing 
locally-densest cuts
using a local search heuristic (see Algorithm~\ref{alg:DCLS}).
This approach is similar to the recursive-sparsest-cut algorithm
of Section~\ref{S:sim:worstcase}. Here, instead of trying to 
solve the densest cut problem (and so being forced 
to use approximation algorithms),
we solve the simpler problem of computing a locally-densest cut.
This yields both a very simple local-search-based algorithm 
and a good approximation guarantee. 

We use the notation $A\oplus x$ to mean the set obtained by adding $x$ to $A$ if $x\notin A$, and by removing $x$ from $A$ if $x\in A$. 
We say that a cut $(A,B)$ is a 
\emph{$\eps/n$-locally-densest cut}
if for any $x$,
$$
  \frac{\cut(A \oplus x, B \oplus x)}{|A\oplus x|\cdot |B\oplus x|} 
  \le \left(1+\frac{\eps}{n}\right) \frac{\cut(A,B)}{|A||B|}.
$$
The following local search algorithm computes an $\eps/n$-locally-densest
cut. 

\begin{algorithm}
  \caption{Local Search for Densest Cut}
  \label{alg:DCLS}
  \begin{algorithmic}[1]
    \State \textbf{Input:} Graph $G=(V,E)$ with edge weights 
    $w: E \mapsto \R_+$
    \State Let $(u,v)$ be an edge of maximum weight
    \State $A \gets \{v\}$, $B \gets V \setminus \{v\}$
    \While{$\exists x$:       $\frac{\cut(A \oplus x, B \oplus x)}{|A\oplus x|\cdot |B\oplus x|} 
      > (1+\eps/n) \frac{\cut(A,B)}{|A||B|}$}
    \State $A \gets A \oplus x$, $B \gets B \oplus x$
\EndWhile \\
    \Return $(A,B)$
  \end{algorithmic}
\end{algorithm}

\begin{theorem}
  \label{T:dis:densestcut}
  Algorithm~\ref{alg:DCLS} computes an $\eps/n$-locally-densest
  cut in time $\tilde{O}(n(n+m)/\eps)$.
\end{theorem}

\begin{proof}
  The proof is 
straightforward and given for completeness.
By definition, the algorithm computes an $\eps/n$-locally densest 
  cut so we only need to argue about the running time.
  The weight of the cut is initially at least $w_{\max}$, the weight
  of the maximum edge weight, and in the end at most $n w_{\max}$.
  Since the weight of the cut increases by a factor of 
  $(1+\eps/n)$ at each iteration, the 
  total number of iterations of the while loop is at most 
  $\log_{1+\eps/n} (n w_{\max}/w_{\max}) = \tilde{O}(n /\eps)$.
  Each iteration takes time $O(m + n)$, 
  so the running time of the algorithm is $\tilde{O}(n (m+n) /\eps)$. 
\end{proof}

\begin{algorithm}
  \caption{Recursive 
    Locally-Densest-Cut  for Hierarchical Clustering}
  \label{alg:greedycut}
  \begin{algorithmic}[1]
    \State \textbf{Input:} Graph $G = (V,E)$, with edge weights
    $w : E \mapsto \R_+$, $\eps > 0$
    \State Compute an $\eps/n$-locally-densest cut $(A,B)$ using
    Algorithm~\ref{alg:DCLS}
    \State Recurse on $G[A]$ and $G[B]$ to obtain 
    rooted trees $T_A$ and $T_B$.
    \State Return the tree $T$ whose root node has two children, $T_A$ and $T_B$. 
  \end{algorithmic}
\end{algorithm}

\begin{theorem}
  \label{T:dis:betterapprox}  
  Algorithm~\ref{alg:greedycut} returns a tree of value at least
  $$\frac{2n}{3} (1-\eps) \sum_{e} w(e) \ge \frac{2}{3} (1-\eps) \opt,$$
  in time $\tilde{O}(n^2(n+m)/\eps)$.
\end{theorem}

The proof relies on the following lemma.
\begin{lemma}
  \label{L:dis:betterapprox}
  Let $(A,B)$ be an $\eps/n$-locally-densest cut. 
  Then, $$(|A|+|B|)\cut(A,B) \ge 2(1-\eps)(|B|w(A) + |A|w(B)).$$
\end{lemma}
\begin{proof}
  Let $v \in A$. By definition of the algorithm,
  $$(1+\eps/n)\frac{\cut(A,B)}{|A||B|} \ge 
  \frac{\cut(A \setminus \{v\}, B\cup\{v\})}{(|A|-1)(|B|+1)}.$$
  Rearranging, 
  $$ 
  \frac{(|A|-1)(|B|+1)}{|A||B|} 
  (1+\eps/n) \cut(A,B) \ge
  \cut(A \setminus \{v\}, B\cup\{v\}) =
   \cut(A,B) + w(v,A) - w(v,B)
  $$
  Summing over all vertices of $A$, we obtain
  $$
  |A|\frac{(|A|-1)(|B|+1)}{|A||B|} 
  (1+\eps/n) \cut(A,B) \ge |A|\cut(A,B) + 2w(A) - \cut(A,B).
  $$
  Rearranging  and simplifying,
  $$
  (|A|-1)(|B|+1)\frac{\eps}{n} \cut(A,B) + 
  (|A|-1)(1+\eps/n) \cut(A,B) \ge 2|B|w(A).
  $$
  Since $|B|+1 \le n$, this gives
  $$
  |A| \cut(A,B) \ge 2(1-\eps)|B|w(A).
  $$
  Proceeding similarly with $B$ and 
  summing the two inequalities
  yields the lemma.
\end{proof}

\begin{proof}[Proof of Theorem~\ref{T:dis:betterapprox}]
  We first show the approximation guarantee.
  We proceed by induction on the number of vertices.  
  The base case is trivial.
  By inductive hypothesis,
  $$\val(T) \ge n\cut(A,B) + 
  \frac{2}{3} \cdot (1-\eps)(|A|w(A) + |B| w(B)),$$ 
  where $n = |A| + |B|$.
  Lemma~\ref{L:dis:betterapprox} implies
  $$n \cut(A,B) = 
  (|A|+|B|)\cut(A,B) \ge 2(1-\eps)(|B|w(A) + |A|W(B)).$$
  Hence, $$\frac{|A|+|B|}{3} \cut(A,B) \ge \frac{2}{3}(1-\eps)(|B|w(A) 
  + |A|w(B)).$$
  Therefore, 
  $$\val(T) \ge  \frac{2n}{3} (1-\eps) (\cut(A,B) + w(A) + w(B))
  = (1-\eps) \frac{2n}{3} \sum_e w(e).$$

  To analyze the running time, observe that
  by Theorem~\ref{T:dis:densestcut}, a recursive call on a graph
  $G'=(V',E')$ takes time $\tilde{O}(|V'|(|V'|+|E'|)/\eps)$ and that the depth of the
  recursion is $O(n)$.
\end{proof}

\begin{remark}
  The average-linkage and the recursive 
  locally-densest-cut algorithms achieve an $O(g_n)$-
  and $O(h_n)$-approximation respectively,
  for any admissible cost function $f$, where 
  $g_n = \max_n f(n)/f(\lceil n/2 \rceil)$.
  $h_n = \max_n f(n)/f(\lceil 2n/3 \rceil)$.
  An almost identical proof yields the result.
\end{remark}

\begin{remark}
  In Section~\ref{ssec:wc-lower-bounds}, we show that
  other commonly used algorithms, such as complete-linkage,
  single-linkage, or bisection $2$-Center, can perform arbitrarily badly
  (see Theorem~\ref{T:badcase:dis:pract}). Hence average-linkage
  is more robust in that sense.
%
\end{remark}

\section{Perfect \StructInputs and Beyond}
\label{S:perfectdata}
\label{sec:perfectdata}
In this section, we focus on \structinputs.
We state that when the input is a perfect \structinput, commonly used algorithms 
(single linkage, average linkage, and complete linkage; as well as some
divisive algorithms -- the bisection $k$-Center and 
sparsest-cut algorithms) yield a  tree of optimal cost, hence (by Definition~\ref{defn:admissibleCF})  a 
ground-truth tree. Some of those results are folklore (and straightforward when there are no ties), but we have been unable to pin down a reference, so we include them for completeness (Section~\ref{section:perfect-is-easy}).
We also introduce a faster optimal algorithm for ``strict''
\structinputs (Section~\ref{S:sub:faster}).
%
%
The proofs
present no difficulty.
The meat of this section is 
Subsection~\ref{S:sub:robust}, where we go beyond \structinputs; 
 we introduce 
$\delta$-adversarially-perturbed
\structinputs and design a simple, more robust algorithm that,  \emph{for any}
admissible objective function, is a $\delta$-approximation.

\subsection{Perfect \StructInputs are Easy}\label{section:perfect-is-easy}

%

\begin{algorithm}
  \caption{Linkage Algorithm for Hierarchical Clustering 
    (similarity setting)}
  \label{alg:linkage4sim}
  \begin{algorithmic}[1]
    \State \textbf{Input:} A graph $G=(V,E)$ with edge weights 
    $w: E \mapsto \R_+$
    \State Create $n$ singleton trees. Root labels: $\calC = \{ \{v_1\}, \ldots, \{ v_n \} \}$
    \State Define $\dist: \calC \times \calC  \mapsto \R_+$:
    $\dist(C_1,C_2) =
    \begin{cases}
      \frac{1}{|C_1||C_2|}      \sum\limits_{x\in C_1,y\in C_2} w((x,y)) &\textbf{Average Linkage}\\
      \min_{x \in C_1, y \in C_2} w((x,y)) & \textbf{Single Linkage} \\
      \max_{x \in C_1, y \in C_2} w((x,y)) & \textbf{Complete Linkage} \\
    \end{cases}$.
    \While{there are at least two trees 
    }
    \State Take the two trees with root labels $C_1,C_2$ such that 
    $\dist(C_1,C_2)$ is maximum
    \State Create a new tree by making those two tree children of 
    a new root node labeled $C_1\cup C_2$ 
    \State Remove $C_1,C_2$ from $\calC$, add $C_1 \cup C_2$ to 
    $\calC$, and update $\dist$
    \EndWhile\\
    \Return the resulting binary tree T
  \end{algorithmic}
\end{algorithm}

In the following, we refer to the \emph{tie breaking} rule of 
Algorithm~\ref{alg:linkage4sim}
as the rule followed by the algorithm for deciding which of 
$C_i,C_j$ or $C_k,C_{\ell}$ to merge, when 
$\max_{C_1,C_2 \in \calC} \dist(C_1,C_2) = \dist(C_i,C_j) = \dist(C_k,C_{\ell})$.

\begin{theorem}\footnote{This Theorem may be folklore, at least when there are no ties, but we have been unable to find a reference.}
  \label{T:linkageAlgs:Perfectdata}
   Assume that the input is 
   a (dissimilarity or similarity)  \structinput. Then, for any admissible objective function,
   the agglomerative heuristics average-linkage, single-linkage, and complete-linkage (see Algorithm~\ref{alg:linkage4sim})
   return an optimal solution. This holds no matter the tie breaking rule of Algorithm~\ref{alg:linkage4sim}.
\end{theorem}
\begin{proof}

We focus on the similarity setting;   the proof for the dissimilarity 
setting is almost identical.
We define the \emph{candidate trees} after $t$ iterations
of the while loop to be sets 
of trees in $\calC$ at that time.
The theorem  follows from
 the following statement, which we will prove by induction on $t$: 
    If $C^t = \{C_1,\ldots,C_k\}$ denotes the set of clusters after
    $t$ iterations, then there exists a generating tree
    $T^t$ for $G$,
    such that the candidate trees are subtrees of $T^t$.

For the base case, initially each candidate tree contains exactly one 
  vertex and the statement holds.
For the general case, let $C_1,C_2$ be the two trees that constitute the $t$$\th$ iteration. 
  By induction, there exists a generating 
  tree $T^{t-1}$ for $G$, and associated weights $W^{t-1}$ (according to Definition~\ref{defn:generating-tree}) such that $C_1$ and
  $C_2$ are subtrees of $T^{t-1}$, rooted at nodes $N^1$ and $N^2$ of $T^{t-1}$ respectively.

 To define $T^t$, we start from $T^{t-1}$. 
  Consider the path $P = \{N^1, N_1,N_2,\ldots,N_k, N^2\}$ joining
  $N^1$ to $N^2$ in $T^{t-1}$ and let $N_r = \lca_{T^{t-1}}(N^1,
  N^2))$.
  If $N_r$ is the parent of $N_1$ and $N_2$, then $T^t=t^{t-1}$, else do the following transformation: remove the subtrees rooted at $N^1$ and at $N_2$; create a new node $N^*$ as second child of $N_k$, and let  $N_1$ and $N_2$ be its children. This defines $T^t$. To define $W^t$, extend $W^{t-1}$ by setting $W^t(N^*)=W(N_r)$.
    \begin{claim}
    \label{Claim:tech:main}
    For any $N_i,N_j \in P$, $W^{t-1}(N_i)= W^{t-1}(N_j)$.
  \end{claim}
Thanks to the inductive hypothesis, with Claim~\ref{Claim:tech:main} it is easy to verify that $W^t$ certifies that $T^t$ is generating for $G$.
\end{proof}

\begin{proof}[Proof of Claim~\ref{Claim:tech:main}]
  Fix a node $N_i$ on the path from $N_r$ to $N^1$ (the 
  argument for nodes on the path from $N_r$ to $N^2$ is similar). 
  By induction 
  $W^{t-1}(N_i) \ge W^{t-1}(N_r)$.
  We show that since the linkage algorithms merge
  the trees $C_1$ and $C_2$, we also have $W^{t-1}(N_i) \le W^{t-1}(N_r)$
  and so $W^{t-1}(N_i) = W^{t-1}(N_r)$, hence the claim.
  Let $w_0 = W^{t-1}(N_r)$.

  By induction, for all $u \in C_1$, $v \in C_2$,
  $w(u,v) = w_0$, and thus $\dist(C_1,C_2) = w_0$ in the execution
  of all the algorithms. 
  Fix a candidate tree $C' \in \calC^t$, $C' \neq C_1,C_2$
  and $C' \subseteq V(N_i)$. 
  Since $\calC$ is 
  a partition of the vertices of the graph and 
  since candidate trees are subtrees of $T^{t-1}$,
  such a cluster exists.
  Thus, for $u \in C_1$, $v \in C'$
  $w(u,v) = 
  W^{t-1}(\lca_{T^{t-1}}(u,v)) = W^{t-1}(N_i) \ge w_0$ since $N_i$ 
  is a descendant of $N_r$.

  It is easy to check that by their definitions, 
  for any of the linkage algorithms, we thus have that 
  $\dist(C_1,C') \ge w_0 = \dist(C_1,C_2)$.
  But since the algorithms merge the clusters at maximum distance,
  it follows that 
  $\dist(C_1,C') \le  \dist(C_1,C_2) = w_0$ and therefore,
  $W^{t-1}(N_i) \le W^{t-1}(N_r)$ and so, $W^{t-1}(N_i) = W^{t-1}(N_r)$
  and the claim follows. This is true no matter the tie breaking
  chosen for the linkage algorithms.
\end{proof}

\paragraph{Divisive Heuristics.}
In this section, we focus on 
two well-known divisive 
heuristics: (1) the bisection $2$-Center which uses a partition-based
clustering objective (the $k$-Center objective)
to divide the input into two 
(non necessarily equal-size) 
parts (see Algorithm~\ref{alg:bisect2centers}), 
and (2) the recursive sparsest-cut algorithm, which can be implemented efficiently for \structinputs (Lemma ~\ref{lemma:sparset-cut-easy}). 

\begin{algorithm}
  \caption{Bisection $2$-Center (similarity setting)}
  \label{alg:bisect2centers}
  \begin{algorithmic}[1]
    \State \textbf{Input:} A graph $G=(V,E)$ and a weight function 
    $w:E \mapsto \R_+$
    \State Find $\{u,v\} \subseteq V$ that maximizes 
    $\min_x \max_{y \in \{u,v\}} w(x,y)$
    \State $A \gets \{x \mid w(x,u) \ge  \max_{y \in \{u,v\}} w(x,y)\}$
    \State $B \gets V \setminus A$.
    \State Apply Bisection $2$-Center on $G[A]$ and $G[B]$ to obtain
    trees $T_A$,$T_B$ respectively\\
    \Return The union tree of $T_A,T_B$.
  \end{algorithmic}
\end{algorithm}

Loosely speaking, 
we show that this algorithm computes an optimal solution
if the optimal solution is unique. More precisely,
for any similarity 
graph $G$, we say that a tree $T$ is \emph{strictly generating}
for $G$ if there exists a weight function $W$ such that 
for any nodes $N_1,N_2$, if $N_1$ appears on the path from $N_2$ to the 
root, then $W(N_1) < W(N_2)$ and for every 
$x, y \in V$, $w(x, y) = W(\lca_{T}(x,y))$.
In this case we say that the input is a strict \structinput.
In the context of dissimilarity, an analogous notion can be defined
and we obtain a similar result.

\begin{theorem}\footnote{This Theorem may be folklore, but we have been unable to find a reference.}
  \label{T:bisectionkcenter}
  For any admissible objective function,
  the bisection $2$-Center algorithm 
  returns an optimal solution for any similarity or dissimilarity 
  graph $G$ that is a strict \structinput.
\end{theorem}

\begin{proof}
  We proceed by induction on the number of nodes in the graph.
  Consider a strictly generating tree $T$ and the corresponding
  weight function $W$.
  Consider the root node $N_r$ of $T$ and let $N_1,N_2$ be the children 
  of the root. Let $(\alpha, \beta)$ be the cut induced by the root 
  node of $T$ (\ie $\alpha = V(N_1)$, $\beta = V(N_2)$). 
  Define $w_0$ to be the weight of an edge
  between $u\in \alpha$ and $v \in \beta$ for any $u,v$ (recall
  that since $T$ is strictly generating all the edges between
  $\alpha$ and $\beta$ are of same weight).
  We show that the bisection $2$-Centers
  algorithm divides the graph into $\alpha$ and $\beta$. Applying the
  inductive hypothesis on both subgraphs yields the result.
  
  Suppose that the algorithm locates the two centers in $\beta$. Then,
  $\min_x \max_{y \in \{u,v\}} w(x,y) = w_0$ since
  the vertices of $\alpha$ are connected
  by an edge of weight $w_0$ to the centers.
  Thus, the value of the clustering is $w_0$. Now, consider a
  clustering consisting of
  a center $c_0$ in $\alpha$ and a center $c_1$ in $\beta$.
  Then, for each vertex $u$, we have 
  $\max_{c \in \{c_0,c_1\}} w(u,c) \ge \min(W(N_1),W(N_2)) > W(N_r) = w_0$ 
  since $T$ and $W$ are strictly generating; Hence a strictly better
  clustering value. Therefore, the algorithm
  locates $x \in \alpha$ and $y \in \beta$. 
  Finally, it is easy to see that the 
  partitioning induced by the centers
  yields parts $A= \alpha$ and $B= \beta$.
\end{proof}

\begin{remark}
To extend our result to (non-strict) \structinputs,
one could consider the 
following variant of the algorithm (which bears
similarities with the popular \emph{elbow method}
for partition-based clustering): Compute a $k$-Center
clustering for all $k \in \{1,\ldots,n\}$ and
partition the graph according to the $k$-Center clustering 
of the smallest $k>1$ 
for which the value of the clustering increases. 
Mimicking the proof
of Theorem~\ref{T:bisectionkcenter}, one can show that
the tree output by the algorithm is generating. 
\end{remark}

We now turn to the recursive sparsest-cut algorithm 
(\ie the recursive $\phi$-sparsest-cut algorithm of 
Section~\ref{S:sim:worstcase}, for $\phi=1$).
The recursive sparsest-cut 
consists in recursively partitioning
the graph according to a sparsest cut of the graph. 
We show (1) that this algorithm yields a tree of optimal cost and
(2) that computing a sparsest cut of a similarity 
graph generated from an ultrametric can be done in linear time.
Finally, we observe that
the analogous algorithm for the dissimilarity setting consists
in recursively partitioning the graph according to the densest
cut of the graph and achieves similar guarantees (and similarly 
the densest cut of a dissimilarity graph
generated from an ultrametric can be computed in linear time).
\begin{theorem}\footnote{This Theorem may be folklore, at least when there are no ties, but we have been unable to find a reference.}
  \label{T:sc:perfect}
  For any admissible objective function,
  the recursive sparsest-cut (respectively densest-cut)
  algorithm computes a tree of optimal cost
  if the input is 
  a similarity (respectively dissimilarity) \structinput.
\end{theorem}

\begin{proof}
  The proof, by induction,  has no difficulty and it may be easier to recreate it than to read it. 
  
  Let $T$ be a generating tree and $W$ be the associated
  weight function.
  Let $N_r$ be the root of $T$, $N_1,N_2$ the children of $N_r$, and $(\alpha= V(N_1), \beta= V(N_2))$  the induced root cut. 
  Since $T$ is strictly generating, all the edges between
  $\alpha$ and $\beta$ are of same weight $w$, which is therefore also the sparsity of $(\alpha,\beta)$.
For every edge $(u,v)$ of the 
  graph, $w(u,v) = W(\lca_T(u,v)) \ge w$, so every cut has sparsity at least $w$, so $(\alpha,\beta)$ has minimum sparsity.
  
  Now, consider the tree $T^*$ computed by the algorithm, and let $(\gamma,\delta)$ denote the sparsest-cut  used by the algorithm at the root (in case of ties it might not different from $(\alpha,\beta)$). 
  By induction the algorithm on $G[\gamma]$ and $G[\delta]$ gives two generating trees $T_{\gamma}$ and $T_{\delta}$
  with associated weight functions $W_{\gamma}$ and
  $W_{\delta}$.  To argue that $T^*$  is generating, we define $W^*$ as follows, where $N^*_r$ denotes the root of $T^*$. 
  $$W^*(N) = 
  \begin{cases}
    W_{\gamma}(N) &\text{if $N \in T_{\gamma}$}\\
    W_{\delta}(N) &\text{if $N \in T_{\delta}$}\\
    w &\text{if $N = N^*_r$}
  \end{cases}
  $$
  By induction $w(u,v) = W(\lca_T(u,v))$ if 
  either both $u,v \in \gamma$, or both $u,v \in \delta$.
  For any $u \in \gamma, v\in \delta$, we have 
  $w(u,v) = w = W(N^*_r) = W(\lca_T(u,v))$.
  Finally, since $w \le w(u,v)$ for any $u,v$, we
  have $W(N^*_r) = w \le W(N)$, for any $N \in T^*$, and therefore
  $T^*$ is generating.
\end{proof}

We then show how to compute a sparsest-cut of a graph that is a
\structinput.

\begin{lemma}\label{lemma:sparset-cut-easy}
If the input graph 
  is a \structinput then the sparsest cut is computed  in $O(n)$ time by the following algorithm: pick an arbitrary vertex $u$, let $w_{\min}$ be the minimum weight of edges adjacent to $u$, and partition $V$ into $A= \{ x \mid w(u,x) > w_{\min} \} $ and $B=V\setminus A$. 
\end{lemma}
\begin{proof}
  Let $w_{\min}  = w(u,v)$. We show that $\cut(A,B)/(|A||B|) = w_{\min}$
  and since $w_{\min}$ is the minimum edge weight of the graph, that
  the cut $(A,B)$ only contains edges of weight $w_{\min}$.
  Fix a generating tree $T$. 
  Consider the path from $u$ to the root of $T$ and let 
  $N_0$ be the first node on the (bottom-up) path such that
  $W(N_0) =w_{\min}$.
  For any vertex $x \in A$, we have that
  $w(u,x) > w_{\min}$. Hence by definition, we have that 
  $N_0$ is an ancestor of $\lca_T(u,x)$.
  Therefore, for any other node $y$ such that $w(u,y) = w_{\min}$,
  we have $\lca_T(u,y) = \lca_T(x,y)$ and so,
  $w(x,y) = W(\lca_T(x,y)) = W(\lca_T(u,y)) = w_{\min}$.
  It follows that all the edges in the cut $(A,B)$ are of weight
  $w_{\min}$ and so, the cut is a sparsest cut.
\end{proof}

\subsection{A Near-Linear Time Algorithm}
\label{S:sub:faster}
In this section, we propose a simple, optimal, 
algorithm for computing a generating tree 
of a \structinput.
For any graph $G$, the running time of this algorithm is $O(n^2)$,
and $\tilde{O}(n)$ if there exists
a tree $T$ that is \emph{strictly generating} for the input.
For completeness we recall that
for any graph $G$, we say that a tree $T$ is strictly generating
for $G$ if there exists a weight function $W$ such that 
for any nodes $N_1,N_2$, if $N_1$ appears on the path from $N_2$ to the 
root, then $W(N_1) < W(N_2)$ and for every 
$x, y \in V$, $w(x, y) = W(\lca_{T}(x,y))$.
In this case we say that the inputs is a \emph{strict \structinput}.

The algorithm is described for the similarity setting but
could be adapted to the dissimilarity case to achieve the
same performances.

\begin{algorithm}
  \caption{Fast and Simple Algorithm for Hierarchical Clustering
  on Perfect Data (similarity setting)}\label{alg:fastsimple}
  \begin{algorithmic}[1]
    \State \textbf{Input:} A graph $G=(V,E)$ and a weight function 
      $w:E \mapsto \R_+$
    \State $p \gets $ random vertex of $V$
    \State Let $w_1 > \ldots > w_k$ be the edge weights of the edges
    that have $p$ as an endpoint
    \State Let $B_i = \{v \mid w(p,v) = w_i\}$, 
    for $1 \le i \le k$.\label{step:alg:2} 

    \State Apply the algorithm recursively on each $G[B_i]$ and obtain a 
    collection of trees $T_1,\ldots,T_k$
    \State Define $T^*_0$ as a tree with $p$ as a single vertex
    \State For any $1 \le i \le k$, define $T^*_i$ to be the union of 
    $T^*_{i-1}$ and $T_i$
    \State Return $T^*_k$
  \end{algorithmic}
\end{algorithm}

\begin{theorem}
  \label{T:alg:fastsimple}
  For any admissible objective function,
  Algorithm~\ref{alg:fastsimple} computes a tree of optimal cost 
  in time $O(n \log^2 n)$
  with high probability if the input is a strict \structinput 
  or in time $O(n^2)$ if the input is a (non-necessarily strict)
  \structinput.
\end{theorem}

\begin{proof}
  We proceed by induction on the number of vertices
  in the graph.
  Let $p$ be the first pivot chosen by the algorithm and
  let $B_1,\ldots,B_k$ be the sets defined by $p$ at 
  Step~\ref{step:alg:2} of the algorithm, with
  $w(p,u) > w(v,p)$, for any $u \in B_i, v\in B_{i+1}$.
  
  We show that for any $u \in B_i, v \in B_{j}$, $j > i$,
  we have $w(u,v) = w(p,v)$.
  Consider a generating tree $T$ and define
  $N_1 = \lca_T(p,u)$ and $N_2 = \lca_T(p,v)$. 
  Since $T,h,\sigma$ is generating and $w(p,u) > w(p,v)$, 
  we have that $N_2$ is an ancestor of $N_1$,
  by Definition~\ref{defn:generating-tree}.
  Therefore, 
  $\lca_T(u,v) = N_2$, and so $w(u,v) = W(N_2) = w(p,v)$.
  Therefore, combining the inductive hypothesis on any $G[B_i]$ and
  by Definition~\ref{defn:generating-tree}
  the tree output by the algorithm is generating.

  A bound of $O(n^2)$ for the running time follows directly from
  the definition of the algorithm.
  We now argue that the running time is $O(n \log^2 n)$ with high 
  probability if the input is strongly generated from a tree $T$.
  First, it is easy to see that a given recursive
  call on a subgraph with $n_0$ vertices takes $O(n_0)$ time.
  Now, observe that if at each recursive call the pivot partitions
  the $n_0$ vertices of its subgraph into buckets of size 
  at most $2n_0/3$, then applying the master theorem implies
  a total running time of $O(n \log n)$.
  Unfortunately, there are trees where picking an arbitrary vertex as
  a pivot yields a single bucket of size $n-1$.
  
  Thus, consider the node $N$ of $T$ that is the first node reached 
  by the walk from the root that always goes to 
  the child tree with the higher number of leaves, stopping when 
  the subtree of $T$ rooted at $N$ contains fewer than $2n/3$ 
  but at least $n/3$ leaves.
  Since $T$ is strongly generating we have that the partition
  into $B_1,\ldots,B_k$ induced by any vertex $v \in V(N)$
  is such that any $B_i$ contains less than $2n/3$ vertices.
  Indeed, for any $u$ such that $\lca_T(u,v)$ is an ancestor 
  of $N$ and $x \in V(N)$, we have that $w(u,v) < w(x,v)$, and so
  $u$ and $x$ belong to different parts of the partition 
  $B_1,\ldots,B_k$.

  Since the number of vertices in $V(N)$ is at least $n/3$,
  the probability of picking one of them is at least $1/3$.
  Therefore, since the pivots are chosen independently, 
  after $c \log n$ recursive calls, the 
  probability of not picking a vertex of $V(N)$ as a pivot is
  $O(1/n^c)$. Taking the union bound yields the theorem.
\end{proof}

\subsection{Beyond Structured Inputs}
\label{S:sub:robust}
Since real-world inputs might sometimes differ 
from our definition of \structinputs introduced 
in the Section~\ref{sec:prelim}, 
we introduce the notion of 
$\delta$-adversarially-perturbed \structinputs.
This notion aims at accounting for noise in the data.
We then design a simple and arguably more reliable 
algorithm (a robust variant of Algorithm~\ref{alg:fastsimple}) 
that achieves a $\delta$-approximation for
$\delta$-adversarially-perturbed \structinputs in
$O(n (n+m))$ time.
An interesting property of this algorithm is that its approximation
guarantee is the same for any admissible objective function.

We first introduce the definition of 
$\delta$-adversarially-perturbed \structinputs.
For any real
$\delta\ge 1$, we say that a weighted graph $G = (V, E, w)$ is a 
\emph{$\delta$-adversarially-perturbed \structinput}
if there exists an ultrametric $(X, d)$, such
that $V \subseteq X$, and for every $x, y \in V, x \neq y$, 
$e = \{x, y\}$
exists, and $ f(d(x, y)) \le w(e) \le \delta f(d(x, y))$, 
where $f : \reals^+ \rightarrow \reals^+$ is a
non-increasing function.  This defines 
$\delta$-adversarially-perturbed \structinputs for similarity graphs 
and an analogous definition applies for dissimilarity graphs.

We now introduce a robust, simple version of 
Algorithm~\ref{alg:fastsimple} that returns a $\delta$-approximation
if the input is a 
$\delta$-adversarially-perturbed \structinputs.
Algorithm~\ref{alg:fastsimple} was partitioning the input graph
based on a single, random vertex. In this slightly more robust
version, the partition is built iteratively:
Vertices are added to the current part if there exists at least
one vertex in the current part or in the parts that were built before
with which they share an edge of high enough weight
(see Algorithm~\ref{alg:robustsimple} for a complete description).

\begin{algorithm}
  \caption{Robust and Simple Algorithm for Hierarchical Clustering
  on $\delta$-adversarially-perturbed \structinputs 
  (similarity setting)}\label{alg:robustsimple}
  \begin{algorithmic}[1]
    \State \textbf{Input:} A graph $G=(V,E)$ and a weight function 
    $w:E \mapsto \R_+$, a parameter $\delta$
    \State $p \gets$ arbitrary vertex of $V$
    \State $i \gets 0$
    \State $\tvi \gets \{p\}$
    \While{$\tvi \neq V$}
    \State Let $p_1 \in \tvi, p_2 \in V \setminus \tilde{V}_i$ s.t.
    $(p_1,p_2)$ is an edge of maximum weight 
    in the cut $(\tilde{V}_i,V \setminus \tilde{V}_i)$
    \State $w_i \gets w(p_1,p_2)$
    \State $B_i \gets \{u \mid w(p_1,u) = w_i\}$
    \While{$\exists u \in V \setminus (\tvi \cup B_i)$ s.t. 
      $\exists v \in B_i \cup \tvi$, $w(u,v) \ge w_i $ 
    }
    \State $B_i \gets B_i \cup \{u\}$.
    \EndWhile
    \State $\tilde{V}_{i+1} \gets \tilde{V}_i \cup B_i$
        \State $i\gets i+1$
    \EndWhile
    \State Let $B_1,\ldots,B_k$ be the sets obtained
    \State Apply the algorithm recursively on each $G[B_i]$ and obtain a 
    collection of trees $T_1,\ldots,T_k$
    \State Define $T^*_0$ as a tree with $p$ as a single vertex
    \State For any $1 \le i \le k$, define $T^*_i$ to be the union of 
    $T^*_{i-1}$ and $T_i$
    \State Return $T^*_k$
  \end{algorithmic}
\end{algorithm}


\begin{theorem}
  \label{T:robustsimple}
  For any admissible objective function,
  Algorithm~\ref{alg:robustsimple} returns a $\delta$-approximation 
  if  the input is a 
  $\delta$-adversarially-perturbed \structinput.
\end{theorem}

To prove the theorem we introduce the following lemma
whose proof is temporarily differed. The lemma states that the
tree built by the algorithm is almost generating (up to a factor
of $\delta$ in the edge weights).
\begin{lemma}
  \label{L:robustsimple}
  Let $T$ be a tree output by Algorithm~\ref{alg:robustsimple},
  let $\calN$ be the set of internal nodes of $T$.
  For any node $N$ with children $N_1,N_2$ 
  there exists a function $\omega: \calN \mapsto \R_+$, such that for any
  $u \in V(N_1), v \in V(N_2)$, $\omega(N) \le w(u,v) 
  \le \delta \omega(N)$.
  Moreover, for any nodes $N,N'$, if $N'$ is an ancestor of $N$,
  we have that $\omega(N) \ge \omega(N')$.
\end{lemma}

Assuming Lemma~\ref{L:robustsimple}, 
the proof of Theorem~\ref{T:robustsimple} is as follows.
\begin{proof}[Proof of Theorem~\ref{T:robustsimple}]
  Let $G=(V,E)$, $w: E \mapsto \R_+$ be the input graph 
  and  $T^*$ be a tree of optimal cost.
  By Lemma~\ref{L:robustsimple}, the tree $T$ output by the
  algorithm is such that for any node $N$ with children $N_1,N_2$ 
  there exists a real $\omega(N)$, such that for any 
  $u \in V(N_1), v \in V(N_2)$, $\omega \le w(u,v) \le \delta \omega$.
  Thus, consider the slightly different input graph
  $G' = (V,E,w')$, where $w' : E \mapsto \R_+$ is defined as follows.
  For any edge $(u,v)$, define $w'(u,v) = \omega(\lca_T(u,v))$.
  Since by Lemma~\ref{L:robustsimple}, for any nodes $N,N'$ of $T$, 
  if $N'$ is an ancestor of $N$,
  we have that $\omega(N) \ge \omega(N')$ and 
  by definition~\ref{defn:generating-tree}, $T$ is 
  generating for $G'$.
  Thus, for any admissible cost function, we have that
  for $G'$, $\cost_{G'}(T) \le \cost_{G'}(T^*)$.

  Finally, observe that for any edge $e$, we have 
  $w'(e) \le w(e) \le \delta w'(e)$. 
  It follows that $\cost_G(T) \le \delta \cost_{G'}(T)$ for 
  any admissible cost function and
  $\cost_{G'}(T^*) \le \cost_{G}(T^*)$.
  Therefore, $\cost_G(T) \le \delta \cost_{G}(T^*) = \delta \opt$.
\end{proof}

\begin{proof}[Proof of Lemma~\ref{L:robustsimple}]
  We proceed by induction on the number of 
  vertices in the graph (the base case is trivial).
  Consider the first recursive call of the algorithm.
  We show the following claim.
  \begin{claim}
    \label{claim:technical:robust}
    For any $1 \le i \le k$, 
    for any $y \in \tvi$, $x \in B_i$,
    $w_i \ge w(x,y) \ge w_i/\delta$. Additionally,
    for any $x,y \in B_i$, $w(x,y) \ge w_i/\delta$.
  \end{claim}
  We first argue that Claim~\ref{claim:technical:robust} implies
  the lemma.
  Let $T$ be the tree output by the algorithm.
  Consider the nodes on the path from $p$ to the root of $T$;
  Let $N_i$ denote the node whose subtree is the union of
  $T^*_{i-1}$ and 
  $T_i$. By definition, $V(T^*_{i-1}) = \tilde{V}_i$ and 
  $V(T_i) = B_i$. Applying Claim~\ref{claim:technical:robust}
  and observing that  $w_i > w_{i+1}$ implies that the lemma holds
  for all the nodes on the path.
  Finally, since for any edge $\{u,v\}$, for $u,v \in B_i$,
  we also have $w(u,v) \ge w_i/\delta$, 
  combining with the inductive hypothesis on $B_i$ 
  implies the lemma for all the nodes of the subtree $T_i$.
\end{proof}
\begin{proof}[Proof of Claim~\ref{claim:technical:robust}]
  Let $(X,d)$ and $f$ be a pair of ultrametric and function that
  is generating for $G$. Fix $i \in \{1,\ldots,k\}$.
  For any vertex $x \in B_i$, let $\sigma(x)$ denote a vertex $y$ that
  is in $\tvi$ or inserted to $B_i$ before $x$ and such that
  $w(y,x) = w_i$. For any vertex $v$, let $\sigma^i(x)$ denotes
  the vertex obtained by applying $\sigma$ $i$ times to $x$ (\ie 
  $\sigma^2(x) = \sigma(\sigma(x))$).
  By definition of the algorithm 
  that for any $x \in B_i$, $\exists s \ge 1$, such that
  $\sigma^s(x) \in \tvi$.
  
  Fix $x \in B_i$.  For any $y \in \tvi$, we have that
  $w(y,x) \le w_i $ 
  since otherwise, the algorithm
  would have added $x$ before.  

  Now, let $y \in \tvi$ or $y$ inserted to $B_i$
  prior to $x$.
  We aim at showing that
  $w(y,x) \ge w_i /\delta$.  Observe that since $X,d$ is an
  ultrametric, $d(x,y) \le \max(d(x,\sigma(x)), d(\sigma(x),y))$.

  We now ``follow'' $\sigma$ by applying the function $\sigma$
  to $\sigma(x)$ and repeating
  until we reach $\sigma(x)^{\ell} = z \in \tvi$.
  Combining with the definition of an ultrametric,
  it follows that
  $$d(x,y) \le \max(d(x,\sigma(x)), d(\sigma(x),\sigma^2(x)),
  \ldots, d(\sigma(x)^{\ell-1}, z), d(z,y)).$$
  If $y$ was in $\tvi$, we define $\hy = y$.
  Otherwise $y$ is also in $B_i$ (and so was added to $B_i$ before
  $x$). We then proceed similarly than for $x$ and 
  ``follow'' $\sigma$. 
  In this case, let $\hy = \sigma^k(y) \in \tvi$.
  Applying the definition of an ultrametric again,
  we obtain
  $$d(x,y) \le \max(d(x,\sigma(x)), d(\sigma(x),\sigma^2(x)),
  \ldots, d(\sigma^{\ell-1}, z),  d(z,\hy), d(y,\sigma(y)), \ldots, 
  d(\sigma^{k-1}(y),\hy)).$$

  Assume for now that
  $d(z,\hy)$ 
  is not greater than the others. 
  Applying the definition of a
  $\delta$-adversarially-perturbed input, we have that
  $$\delta w(x,y) \ge \min(\ldots, w(\sigma^a(x),\sigma^{a+1}(x)), 
  \ldots, w(\sigma^b(y),\sigma^{b+1}(y)), \ldots).$$ 
  Following the definition of $\sigma$,
  we have $w(v,\sigma(v)) \ge w_i$,
  $\forall v$. Therefore, we conclude 
  $\delta w(x,y) \ge w_i$.  

  We thus turn to the case where $d(z,\hy)$ is greater than the
  others.
  Since both $z,\hy \in \tvi$, we have that they belong
  to some $B_{j_0},B_{j_1}$, where $j_0,j_1 < i$. 
  We consider the minimum $j$
  such that a pair at distance at least $d(z,\hy)$ was added to
  $\tilde{V}_j$.  
  Consider such a pair
  $u,v \in \tilde{V}_j$ satisfying $d(u,v) \ge d(z,\hy)$ and suppose
  w.l.o.g that $v \in B_{j-1}$ (we could have 
  either $u \in B_{j-1}$ or $u \in \tilde{V}_{j-1}$). Again, we follow
  the path $\sigma(v), \sigma(\sigma(v)),\ldots$, until we reach
  $\sigma^{r_1}(v) \in \tilde{V}_{j-1}$ and similarly for $u$:
  $\sigma^{r_2}(u) \in \tilde{V}_{j-1}$. Applying the 
  definition of an ultrametric this yields that 
  \begin{equation}
    \label{eq:claim:ultram}
    d(u,v) \le \max(\ldots, d(\sigma^a(u),\sigma^{a+1}(u)), \ldots,
    d(\sigma^b(v),\sigma^{b+1}(v)), \ldots, d(\sigma^{r_1}(v), 
    \sigma^{r_2}(u))).
  \end{equation}

  Now the difference is that $\tilde{V}_{j-1}$
  does not contain any pair at distance at least $d(z,\hy)$.
  Therefore, we have
  $d(\sigma^{r_1}(v), \sigma^{r_2}(u)) < d(z,\hy)$.
  Moreover, recall that
  by definition of $u,v$,
  $d(z,\hy) \le d(u,v)$.
  Thus, $d(\sigma^{r_1}(v),\sigma^{r_2}(u))$ is not 
  the maximum in Equation~\ref{eq:claim:ultram} since
  it is smaller than the left-hand side.
  Simplifying Equation~\ref{eq:claim:ultram}
  yields
  $$
  d(x,y) < d(z,\hy) \le d(u,v) 
  \le \max(\ldots, d(\sigma^a(u),\sigma^{a+1}(u)), \ldots,
  d(\sigma^b(v),\sigma^{b+1}(v)), \ldots).
  $$

  By definition of a
  $\delta$-adversarially-perturbed input, we obtain
  $\delta w(x,y) \ge \min_{\ell} w(\sigma^{\ell}(b),
  \sigma^{\ell+1}(b)) \ge w_j$.  Now, it is easy to see that
  for $j < i$,
  $w_i < w_j$ and therefore $\delta w(x,y) \ge w_i$.
  
  We conclude that for any $y \in \tvi$, $x \in B_i$,
  $w_i \ge w(x,y) \ge w_i/\delta $ and
  for $x,y \in B_i$,
  we have that $w(x,y) \ge w_i/\delta$, as claimed.  
\end{proof}

\section{Worst-Case Analysis of Common Heuristics}
\label{ssec:wc-lower-bounds}
The results presented in this section shows 
that for both the similarity and dissimilarity settings, 
some of the widely-used heuristics may perform badly.
The proofs are not difficult nor particularly interesting, but the results stand in sharp contrast to structured inputs and help motivate our study of inputs beyond worst case.

\paragraph{Similarity Graphs.}
We show that for very simple input graphs (\ie unweighted trees), 
the linkage algorithms (adapted to 
the similarity setting, see Algorithm~\ref{alg:linkage4sim}) 
may perform badly. 

\begin{theorem}
  \label{T:sim:hard:sandc}
  There exists an infinite family of inputs on which the single-linkage
  and complete-linkage algorithms output a solution of cost 
  $\Omega(n \opt/\log n)$.
\end{theorem}

\begin{proof}
  The family of inputs consists of the graphs that represent paths
  of length $n>2$.
  More formally, Let $G_n$ be a graph on $n$ vertices such that
  $V = \{v_1,\ldots,v_n\}$ and that has the following edge weights.
  Let $w(v_{i-1},v_i) = w(v_i,v_{i+1}) = 1$, for all $1 < i < n$ and
  for any $i,j$, $j \notin \{i-1,i,i+1\}$, define $w(v_i,v_j) = 0$.
  \begin{claim}
    \label{cl:wc:nlogn}
    $\opt(G_n) = O(n \log n)$.
  \end{claim}
  \begin{proof}
    Consider the tree $T^*$ that recursively divides the path into two
    subpaths of equal length.
    We aim at showing that $\cost(T^*) = O(n \log n)$.
    The cost induced by the root node is $n$ (since there is
    only one edge joining the two subpaths). The cost induced
    by each child of
    the root is $n/2$ since there is again only one edge joining the
    two sub-subpaths and now only $n/2$ leaves in the two subtrees.
    A direct induction shows that for a descendant at distance 
    $i$ from the root, the cost induced by this node is $n/2^i$.
    Since the number of children at distance 
    $i$ is $2^i$, we obtain that the total cost induced by all
    the children at distance $i$ is $n$.
    Since the tree divides the graph into two subgraph of equal size,
    there are at most $O(\log n)$ levels and therefore, the total cost
    of $T$ (and so $\opt(G_n)$) is $O(n \log n)$.
  \end{proof}
  \textbf{Complete-Linkage.}
  We show that the complete-linkage algorithm
  could perform a sequence of merges that would induce a 
  tree of cost $\Omega(n^2)$.
  At start, each cluster contains a single vertex and so, 
  the algorithm could merge any two clusters $\{v_i\}$, $\{v_{i+1}\}$
  with $1\le i<n$ since their distance are all 1 (and it is the
  maximum edge weight in the graph).
  Suppose w.l.o.g that the algorithm merges $v_1,v_2$.
  This yields a cluster $C_1$ such that 
  the maximum distance between vertices of $C_1$ and $v_3$ is 1.
  Thus, assume w.l.o.g that the second merge of the algorithm
  is $C_1,v_3$.
  Again, this yields a cluster $C_2$ whose maximum distance to $v_4$
  is also 1. A direct induction shows that the algorithm output a 
  tree whose root induces the cut $(V \setminus \{v_n\}, v_n)$
  and one of the child induces the cut 
  $(V \setminus \{v_{n-1},v_n\}, v_{n-1})$ and so on.
  We now argue that this tree $\hat T$ has cost $\Omega(n^2)$.
  Indeed, for any $1< i\le n$, we have 
  $V(\lca_{\hat T}(v_{i-1},v_i)) = i$. Thus the cost is 
  at least $\sum_{i = 2}^n i = \Omega(n^2)$.

  \textbf{Single-Linkage.}
  We now turn to the case of the single-linkage algorithm.
  Recall that the algorithm merges the two candidate clusters $C_i,C_j$
  that minimize $w(u,v)$ for $u \in C_i$, $v\in C_j$.

  At start, each cluster contains a single vertex and so, 
  the algorithm could merge any two clusters $\{v_i\}, \{v_j\}$ for 
  $j \notin \{i-1,i,i+1\}$ since the edge weight is 0 (and it
  is the minimum edge weight).
  Suppose w.l.o.g that the algorithm merges $v_1,v_3$.
  This yields a cluster $C_1$ such that 
  the distance between vertices of $C_1$ and any $v_i$ for
  $i = 1\mod 2$ is 0.
  Thus, assume w.l.o.g that the second merge of the algorithm
  is $C_1,v_5$. 
  A direct induction shows that w.l.o.g the output tree contains
  a node $\tilde N$ such that $V(\tilde N)$ contains all the 
  vertices of odd indices.
  Now observe that the cost of the tree is at least
  $|V(\tilde N)|\cdot \cut(V(\tilde N), V \setminus V(\tilde N)) 
  = \Omega(n^2)$.

  Thus, by Claim~\ref{cl:wc:nlogn} the single-linkage and 
  complete-linkage algorithms can output a solution of cost 
  $\Omega(n \opt/\log n)$.
\end{proof}

\begin{theorem}
  \label{T:sim:hard:avg}
  There exists an infinite family of inputs on which the average-linkage
  algorithm output a solution of cost $\Omega(n^{1/3} \opt)$.
\end{theorem}

\begin{proof}
  For any $n = 2^i$ for some integer $i$, 
  we define a tree $T_n = (V,E)$
  as follows. Let $k= n^{1/3}$.
  Let $P = (u_1,\ldots,u_k)$ be a path of length $k$
  (\ie for each $1 \le i < k$, 
  we have an edge between $u_i$ and $u_{i+1}$).
  For each $u_i$, we define a collection 
  $P_i = \{P^i_1 = (V^i_1,E^i_1),\ldots, P^i_k = (V^i_k,E^i_k) \}$ 
  of $k$ paths of length $k$ and for each $P^i_j$ we connect 
  one of its extremities to $u_i$.
  Define $V_i = \{u_i\} \bigcup_j V^i_j$.

  \begin{claim}
    \label{cl:prop:avg}
    $\opt(T_n) \le 3n^{4/3}$
  \end{claim}
  \begin{proof}
    Consider the following non-binary solution tree $T^*$:
    Let the root have children $N_1,\ldots,N_k$ such that 
    $V(N_i) = V_i$ and for each child $N_i$
    let it have children $N^j_i$ such that $V(N^j_i) = V^j_i$.
    Finally, for each $N^j_i$ let the subtree rooted at $N^j_i$ 
    be any tree.

    We now analyze the cost of $T^*$.
    Observe that for each edge $e$ in the path $P$, we have 
    $|V(\lca_{T^*}(e))| = n$.
    Moreover, for each edge $e$ connecting a path $P^i_j$ to 
    $u_i$, we have $|V(\lca_{T^*}(e))| = k^2 = n^{2/3}$.
    Finally, for each edge $e$ whose both endpoints are in 
    a path $P^i_j$, we have that 
    $|V(\lca_{T^*}(e))| \le k = n^{1/3}$.
    
    We now sum up over all edges to obtain the overall cost of $T_n$.
    There are $k = n^{1/3}$ edges in $P$; They incur a cost of 
    $nk = n^{4/3}$.
    There are $k^2$ edges joining a vertex $u_i$ to a path $P^i_j$;
    They incur a cost of $k^2 \cdot n^{2/3} = n^{4/3}$.
    Finally, there are $k^3$ edges whose both endpoints are in
    a path $P^i_j$; They incur a cost of $k^3 \cdot n^{1/3} \le n^{4/3}$.
    Thus, the total cost of this tree is at 
    most $3 n^{4/3} \ge \opt(T_n)$.
  \end{proof}

  We now argue that there exists a sequence of merges done
  by the average-linkage algorithm that yield a solution of 
  cost at least $n^{5/3}$.
  \begin{claim}
    \label{cl:prop:avg2}
    There exists a sequence of merges and an
    integer $t$ such that the candidate trees at time $t$
    have leaves sets $\{\{u_1,\ldots,u_k\}\} \bigcup_{i,j} \{V^i_j\}$.
  \end{claim}
  Equipped with this claim, we can finish the proof of the proposition.
  Since there is no edge between $V^i_j$ and $V^{i'}_{j'}$ for $i' \neq i$
  or $j' \neq j$ the distance between those trees in the algorithm
  will always be 0. However, the distance between the tree $\tT$ that
  has leaves set $\{u_1,\ldots,u_k\}$ and 
  any other tree is positive (since there
  is one edge joining those two sets of vertices in $T_n$).
  Thus, the algorithm will merge $\tT$ with some 
  tree whose vertex set is exactly $V^i_j$ for some $i,j$.
  For the same reasons, the resulting cluster will be merged to a cluster
  whose vertex set is exactly $V^{i'}_{j'}$, and so on. 
  Hence,
  after $n/2k = k^2/2$ such merges, the tree $\tT$ has a leaves set
  of size  
  $k \cdot k^2/2 = n/2$. However, the number of edges from this cluster
  to the other candidate clusters is $k^2/2$ (since the other remaining
  clusters corresponds to vertex sets $V^i_j$ for some $i,j$).
  For each such edge $e$ we have 
  $|V(\lca_T(e))| \ge n/2$. Since there are $k^2/2$ of them, 
  the resulting tree has cost $\Omega(n^{5/3})$. Combining with 
  Claim~\ref{cl:prop:avg} yields the theorem.
\end{proof}  
  We thus turn to the proof of Claim~\ref{cl:prop:avg2}.
  \begin{proof}[Proof of Claim~\ref{cl:prop:avg2}]
    Given a graph $G$ a set of candidate trees $\calC$, define 
    $G/\calC$ to be the graph resulting from the contraction
    of all the edges whose both endpoints belong to the same cluster.
    We show a slightly stronger claim. We show that for any graph
    $G$ and candidate trees $\calV$ such that
    \begin{enumerate}
    \item All the candidate clusters in $\calV$ have the same size; and
    \item There exists a bijection $\phi$ between vertices $v \in T_n$ 
      and vertices in $G/\calC$;
    \end{enumerate}
    There exists a sequence of merges and an
    integer $t$ such that the candidate trees at time $t$
    have leaves sets $\{\{\phi(u_1),\ldots,\phi(u_k)\}\} 
    \bigcup_{i,j} \{\phi(V^i_j)\}$ where $\phi(V^i_j) = 
    \{\phi(v) \mid v \in V^i_j \}$.
    
    This slightly stronger statement yields the claim by observing 
    that $T_n$ and the candidate trees at the start of the algorithm
    satisfies the conditions of the statement.

    We proceed by induction on the number of vertices of the graph.
    Let $V^i_j = \{v^i_j(1),\ldots,v^i_j(k)\}$ such that
    $(v^i_j(\ell),v^i_j(\ell+1)) \in E^i_j$ for any $1 \le \ell < k$,
    and $(v^i_j(k), u_i) \in E$.

    We argue that the algorithm could perform a sequence of merges
    that results in the following set $\calC$
    of candidate trees.
    $\calC$ contains candidate trees 
    $U^{i} = \phi(u_{2i-1}) \cup \phi(u_{2i})$ for 
    $1 \le i < k/2$, and
    for each  $i,j$, candidate trees 
    $v_{i,j,\ell} = \phi(v^i_j(2\ell-1) \cup \phi(v^i_j(2\ell))$, 
    for $1\le \ell <k/2$. Let $s_0$ be the number of vertices
    in each candidate tree.

    At first, all the trees contain a single vertex and so, 
    for each adjacent vertices of the graph the distance between
    their corresponding trees in the algorithm is $1/s_0$. For any
    non-adjacent pair of vertices, the corresponding
    trees are at distance 0.
    Thus, w.l.o.g assume the algorithm first merges $u_1,u_2$.
    Then, the distance between the newly created tree $U^1$ and 
    any other candidate tree $C$ is 0 if there is no edge between
    $u_1$ and $u_2$ and $C$ or $1/(2s_0)$ if there is one 
    (since $U^1$ contains
    now two vertices). For the other candidate trees the distance
    is unchanged.
    Thus, the algorithm could merge vertices
    $u_3,u_4$. Now, observe that the distance between $U^2$ and 
    $U^1$ is at most $1/(4s_0)$. 
    Thus, it is possible to repeat the argument
    and assume that the algorithm merges the candidate trees 
    corresponding to $u_5,u_6$.
    Repeating this argument $k/2$ times yields that after $k/2$ merges,
    the algorithm has generated the candidates trees
    $U_1,\ldots,U_{k/2-1}$. The other candidate trees still contain
    a single vertex.
    Thus, the algorithm is now forced to merge candidate trees
    that contains single vertices that are adjacent (since their
    distance is $1/s_0$ and any other distance is $< 1/s_0$).
    Assume, w.l.o.g, that the algorithm merges $v^1_1(1), v^1_1(2)$.
    Again, applying a similar reasoning to each 
    $v^1_1(2\ell-1),v^1_1(2\ell)$ yields the set of candidate clusters 
    $v_{1,1,1},\ldots,v_{1,1,k/2-1}$. Applying this argument to all sets
    $V^i_j$ yields that the algorithm could perform a sequence of merges
    that results in the set $\calC$ of candidate clusters 
    described above.

    Now, all the clusters have size $2s_0$ and there exists 
    a bijection between vertices of $G/\calC$ and $T_{n/2}$.
    Therefore, combining with the induction hypothesis 
    yields the claim.
  \end{proof}

\paragraph{Dissimilarity Graphs.}
We now show that single-linkage, complete-linkage, 
and bisection $2$-Center
might return a solution that is arbitrarily bad compared to $\opt$
in some cases.
Hence, since average-linkage achieves a
$2$-approximation in the worst-case it seems that
it is more robust than the other
algorithms used in practice. 
 
\begin{theorem}
  \label{T:badcase:dis:pract}
  For each of the single-linkage, 
  complete-linkage, and bisection $2$-Center algorithms,
  there exists a family of inputs for which the
  algorithm outputs a solution of value $O(\opt/n)$.
\end{theorem}

\begin{proof}
  We define the family of inputs as follow.
  For any $n>2$, the graph $G_n$ consists of $n$ vertices 
  $V = \{v_1,\ldots,v_{n-1}, u\}$ and the edge weights are the 
  following: For any $i,j \in \{1,\ldots,n-1\}$, $w(v_i,v_j) = 1$,
  for any $1 < i \le n- 1$, $w(v_i, u) = 1$, and $w(v_1,u) = W$
  for some fixed $W \ge n^3$.
  Consider the tree $T^*$ whose root induces a cut 
  $(V \setminus \{u\},\{u\})$. Then, the value of this tree
  (and so $\opt$) is at least $n W$, since $|V(\lca_{T^*}(v_1,u))| = n$.
  
  \textbf{Single-Linkage.}
  At start, all the clusters are at distance 1 from each other 
  except $v_1$ and $u$ that are at distance $W$. Thus, suppose that 
  the first merge generates a candidate tree $C_1$
  whose leaves set is $\{v_1,v_2\}$. 
  Now, since $w(v_2,u) = 1$, we have that all the clusters
  are at distance 1 from each other.
  Therefore, the next merge could possibly generate the cluster
  $C_2$ with leaves sets $\{u,v_1,v_2\}$. 
  Assume w.l.o.g that this is the case 
  and let $T$ be the tree output by the algorithm.
  We obtain $|V(\lca_T(u, v_1))| = 3$ and so, since for any
  $v_i,v_j$, $|V(v_i,v_j)| \le n$,
  $\val(T) \le n^2+ 3W \le 4W$, since $W > W^3$.
  Hence, $\val(T) = O(\val(T^*)/n)$.  

  \textbf{Complete-Linkage.}
  Again, at first all the clusters are at distance 1 from each other 
  except $v_1$ and $u$ that are at distance $W$. Since the algorithm 
  merges the two clusters that are at maximum distance, it merges
  $u$ and $v_1$. Again, let $T$ be the tree output by the algorithm.
  We have $\val(T) \le n^2+ 2W \le 3W$, since $W > W^3$.
  Hence, $\val(T) = O(\val(T^*)/n)$.

  \textbf{Bisection $2$-Center.}
  It is easy to see that for any location of the two centers, 
  the cost of the clustering is $1$.
  Thus, suppose that the algorithm locates centers in $v_2,v_3$
  and that the induced partitioning is 
  $\{v_1,v_2,u\}, V\setminus \{v_1,v_2,u\}$.
  It follows that $|V(\lca_T(u, v_1))| \le 3$ and so,
  $\val(T) \le n^2+ 3W \le 4W$, since $W > W^3$.
  Again, $\val(T) = O(\val(T^*)/n)$.  
\end{proof}

\begin{proposition}
\label{prop:dis:ms}
  For any input $\calI$ lying in a metric space, 
  for any solution tree $T$ for $\calI$, we have $\val(T) = O(\opt)$.
\end{proposition}
\begin{proof}
  Consider a solution tree $T$ and the node $u_0$ of $T^*_S$ 
  that is the first node reached by the walk from the root 
  that always goes to 
  the child tree with the higher number of leaves, stopping when the 
  subtree of $T^*_S$ rooted at $u_0$ contains fewer than $2n/3$ leaves.
  Let $A = V(u_0), B = V \setminus V(u_0)$.
  Note that the number of edges in $G[A]$ is at 
  most $a = {|A| \choose 2}$, the number of edges in $G[B]$ is at 
  most $b = {|B| \choose 2}$, whereas the number of 
  edges in the cut $(A,B)$ 
  is $|A| \cdot |B|$. Recall that $n/3 \le |A|,|B| \le 2n/3$ and
  so $a,b = \Theta(|A| \cdot |B|)$.
  Finally observe that for each edge $(u,v) \in G[A]$, we have
  $w(u,v) \le w(u,x) + w(x,v)$ for any $x \in B$. 
  Thus, since $a+b = \Theta(|A| \cdot |B|)$, by a simple counting 
  argument, we deduce
  $\val(T) = \Omega(n \sum_e w(e))$ and by Fact~\ref{F:optcost},
  $\Omega(\opt)$.
\end{proof}

\subsubsection*{Acknowledgments} 
The authors are grateful to Sanjoy Dasgupta for sharing thoughtful comments at
various stages of this project.

\bibliographystyle{plainnat}
\bibliography{biblio.bib}

\end{document}